\documentclass[11pt,onecolumn]{IEEEtran}
\usepackage[margin=.7in]{geometry}
\usepackage[font=small]{caption}
\usepackage{amsmath}
\usepackage{amssymb}
\usepackage{amsfonts}
\usepackage{amsthm}
\usepackage{setspace}
\usepackage{cite}
\usepackage{ifthen}
\usepackage{epsfig}
\usepackage{array}
\usepackage{enumerate}
\usepackage{subfig}
\usepackage{algorithm}
\usepackage[all]{xy}
\usepackage{algpseudocode}
\usepackage{times}
\usepackage{tikz}
\usepackage{paralist}
\usepackage{float}
\usepackage{dblfloatfix}
\usetikzlibrary{shapes,positioning,arrows,automata}

\newtheorem{thm}{Theorem}
\newtheorem{theorem}[thm]{Theorem}
\newtheorem{proposition}[thm]{Proposition}
\newtheorem{lem}[thm]{Lemma}

\newtheorem{example}{Example}
\newtheorem{definition}{Definition} 

\newcommand{\request}{r}
\newcommand{\beq}{\begin{equation}}
\newcommand{\eeq}{\end{equation}}
\newcommand{\bea}{\begin{eqnarray}}
\newcommand{\eea}{\end{eqnarray}}
\newcommand{\bean}{\begin{eqnarray*}}
\newcommand{\eean}{\end{eqnarray*}}
\newcommand{\bit}{\begin{itemize}}
\newcommand{\eit}{\end{itemize}}
\newcommand{\ben}{\begin{enumerate}}
\newcommand{\een}{\end{enumerate}}
\newcommand{\blem}{\begin{lem}}
\newcommand{\elem}{\end{lem}}
\newcommand{\bthm}{\begin{thm}}
\newcommand{\ethm}{\end{thm}}
\newcommand{\bpf}{\begin{IEEEproof}}
\newcommand{\epf}{\end{IEEEproof}}

\newcommand{\comment}[1]{}
\newcommand{\supth}{^{\textrm{th}}}
\newcommand{\On}{\State\hspace{-.35cm} \textbf{On} }

\newcommand{\spaceBetweenScheduling}{1cm}

\emergencystretch=\maxdimen %to force word wrapping

\title{\vspace{.25in}When Do Redundant Requests Reduce Latency ?}
\author{Nihar B. Shah, Kangwook Lee, Kannan Ramchandran\\Department of Electrical Engineering and Computer Sciences,\\University of California at Berkeley\\\{nihar,\,kw1jjang,\,kannanr\}@eecs.berkeley.edu
}

\begin{document}
\maketitle
\thispagestyle{empty}

\begin{abstract}
Several systems possess the flexibility to serve requests in more than one way. For instance, a distributed storage system storing multiple replicas of the data can serve a request from any of the multiple servers that store the requested data, or a computational task may be performed in a compute-cluster by any one of multiple processors. In such systems, the latency of serving the requests may potentially be reduced by sending \textit{redundant requests}: a request may be sent to more servers than needed, and it is deemed served when the requisite number of servers complete service. Such a mechanism trades off the possibility of faster execution of at least one copy of the request with the increase in the delay due to an increased load on the system. Due to this tradeoff, it is unclear when redundant requests may actually help. Several recent works empirically evaluate the latency performance of redundant requests in diverse settings.%, however, to the best of our knowledge, no analytical characterization of the situations in which redundant requests will help is known.

This work aims at an analytical study of the latency performance of redundant requests, with the primary goals of characterizing under what scenarios sending redundant requests will help (and under what scenarios they will not help), as well as designing optimal redundant-requesting policies. We first present a model that captures the key features of such systems. We show that when service times are i.i.d. memoryless or ``heavier'', and when the additional copies of already-completed jobs can be removed instantly, redundant requests reduce the average latency. On the other hand, when service times are ``lighter'' or when service times are memoryless and removal of jobs is not instantaneous, then not having any redundancy in the requests is optimal under high loads. Our results hold for arbitrary arrival processes.
\end{abstract}

\section{Introduction}
Several systems possess the flexibility to serve requests in more than one way. For instance, in a cluster with $n$ processors, a computation may be performed at any one of the $n$ processors; in a distributed storage system where data is stored using an $(n,k)$ Reed-Solomon code, a read-request may be served by reading data from any $k$ of the $n$ servers; in a network with $n$ available paths from the source to the destination, communication may be performed by transmitting across any one of the $n$ paths. % in a distributed storage system using a Product-Matrix code~\cite{ourProductMatrix_supershort}, the data of any one server can be recovered from any subset of servers of a fixed size. 
In such settings, the latency of serving the requests can potentially be reduced by sending \textit{redundant requests}. Under a policy of sending redundant requests, each request is attempted to be served in more than one way. The request is deemed served when it is served in any one of these ways. Following this, the other copies of this request may be removed from the system. 

It is unclear whether or not such a policy of having redundant requests will actually reduce the latency. On one hand, for any individual request, one would expect the latency to reduce since the time taken to process the request is the \textit{minimum} of the processing times of its multiple copies. On the other hand, introducing redundancy in the requests consumes additional resources and increases the overall load on the system, thereby adversely affecting the latency.

Many recent works such as~\cite{snoeren2001mesh,andersen2005improving,pitkanen2007redundancy,han2011rpt,ananthanarayanan2012let,huang2012erasure,vulimiri2012more,dean2013tail,liang2013fast,stewart2013zoolander,flach2013reducing} perform empirical studies on the latency performance of sending redundant requests, and report reductions in latency in several scenarios (but increases in some others). However, despite a significant interest among practitioners, to the best of our knowledge, no rigorous analysis is known as to when redundant requests help in reducing latency (and when not). This precisely forms the goal of this work. We consider a model based on the `MDS queue' model~\cite{MDSqueue}, which captures some of the key features of such systems, and can serve as a building block for more complex systems. %Under this model, we show that when service times are i.i.d. with a distribution that is memoryless or \textit{heavy-everywhere}, and when removing unfinished jobs incurs negligible costs, redundant requests necessarily reduce the latency for any choice of system parameters and any arrival process. However, when the load on the system is high, the opposite holds true for distributions that are \textit{light-everywhere}. These results hold when the system has a centralized buffer as well as when the buffers are distributed. %We also provide insightful simulations of more general settings.
Under this model, for several classes of distributions of the arrival, service and removal times, we derive the optimal redundant-requesting policies. These results are summarized in Table~\ref{tab:summary_of_results}. % Certain notation and terms employed in the table, such as `heavy-everywhere' or `light-everywhere' distributions, shall be defined subsequently in the paper.
Our proof techniques allow for arbitrary arrival sequences and are \textit{not} restricted to (asymptotic) steady-state settings.

\begin{table*}
\centering
\begin{tabular}{|l|l|l|l|l|c|l||l|c|}
\hline
$n$&$k$&\shortstack{Arrival\\process}&\shortstack{Service\\distribution}& Buffers&\shortstack{Removal\\cost}&Load&\textbf{Optimal policy}&Theorem \#\\
\hline
\hline
any&1&any&i.i.d., memoryless & centralized & 0 & any & send to all&\ref{thm:rep_flood}\\
any&any&any& i.i.d., memoryless&centralized&0&any&send to all&\ref{thm:memoryless_general_k}\\
any & 1 & any & i.i.d., heavy-everywhere & centralized & 0 & high & send to all&\ref{thm:heavier_flood}\\
any & 1 & any & i.i.d., light-everywhere & centralized & any & high & no redundancy&\ref{thm:lighter_noflood}\\
any & 1 & any & i.i.d., memoryless & centralized & $>$0 & high & no redundancy&\ref{thm:removalCosts_memless}\\
any&any&any&i.i.d., memoryless & distributed & 0 & any & send to all&\ref{thm:memoryless_general_k_distributed}\\
any & 1 & any & i.i.d., heavy-everywhere & distributed & 0 & high & send to all&\ref{thm:heavy_distributed}\\
any & 1 & any & i.i.d., light-everywhere & distributed & any & high & no redundancy&\ref{thm:lighter_noflood_distributed}\\
any & 1 & any & i.i.d., memoryless & distributed & $>$0 & high & no redundancy&\ref{thm:removalCosts_memless_distributed}\\
\hline
\end{tabular}
\caption{Summary of analysis of \textit{when} redundant requests reduce latency, and the optimal policy of redundant-requesting under various settings. The `heavy-everywhere' and `light-everywhere' classes of distributions are defined subsequently in Section~\ref{sec:analytical}. An example of a heavy-everywhere distribution is a mixture of exponential distributions; two examples of light-everywhere distributions are an exponential distribution shifted by a constant, and the uniform distribution. By `high load' we mean a $100\%$ utilization of the servers.}
\label{tab:summary_of_results}
\end{table*}

The remainder of this paper is organized as follows. Section~\ref{sec:literature} discusses related literature. Section~\ref{sec:model} presents the system model with a centralized buffer. Section~\ref{sec:analytical} presents analytical results for such a centralized setting. Section~\ref{sec:model_distributed} describes a distributed setting where each server has its own buffer. Section~\ref{sec:analytical_distributed} presents analytical results under this distributed setting. Finally, Section~\ref{sec:conclusion} presents conclusions and discusses open problems. Appendix~\ref{app:everywhere_dist} presents properties and examples of the \textit{heavy-everywhere} and \textit{light-everywhere} distributions defined in the paper. Appendix~\ref{app:proofs} contains proofs of all the analytical results.

\section{Related Literature}\label{sec:literature}
Policies that try to reduce latency by sending redundant requests have been previously studied, largely empirically, in~\cite{snoeren2001mesh,andersen2005improving,pitkanen2007redundancy,han2011rpt,ananthanarayanan2012let,huang2012erasure,vulimiri2012more,dean2013tail,liang2013fast,stewart2013zoolander,flach2013reducing}. These works evaluate system performance under redundant requests for several applications, and report reduction in the latency in many cases. For instance, Ananthanarayanan et al.~\cite{ananthanarayanan2012let} consider the setting where requests take the form of computations to be performed at processors. In their setting, requests have diffferent workloads, and the authors propose adding redundancy in the requests with lighter workloads. They observe that on the PlanetLab network, the average completion time of the requests with lighter workloads improves by 47\%, at the cost of just 3\% extra resources. Huang et al.~\cite{huang2012erasure} consider a distributed stoage system where the data is stored using an $(n=16,\ k=12)$ Reed-Solomon code. For $k' \in \{12,13,14,15\}$, they perform the task of decoding the original data by connecting to $k'$ of the nodes and decoding from the $k$ pieces of encoded data that arrive first. They empirically observe that the latency reduces upon increase in $k'$. In a related setup, codes and algorithms tailored specifically for employing redundant requests in distributed storage are designed in~\cite{rashmi2012errors} for latency-sensitive settings, allowing for data stored in a busy or a failed node to be obtained by downloading little chunks of data from other nodes. In particular, these codes provide the ability to connect to more nodes than required and use the data received from the first subset to respond, treating the other slower nodes as \textit{erasures}. Vulimiri et al.~\cite{vulimiri2012more} propose sending DNS queries to multiple servers. They observe that on PlanetLab servers, the latency of the DNS queries reduces with an increase in the number of DNS servers queried. Dean and Barroso~\cite{dean2013tail} observe a reduction in latency in Google's system when  requests are sent to two servers instead of one. Liang and Kozat~\cite{liang2013fast} perform experiments on the Amazon EC2 cloud. They observe that when the rate of arrival of the requests is low, the latency reduces when the requests are sent to a higher number of servers. However, when the rate of arrival is high, they observe that a high redundancy in the requests increases the latency.

\begin{algorithm*}[t!]
\begin{algorithmic}
\On arrival of a request (``batch'')
\State divide the batch into $\request$ jobs
\State assign as many jobs (of the new batch) as possible to idle servers
\State append the remaining jobs (if any) as a new batch at the end of the buffer 
~\\
\On completion of processing by a server (say, server $s_0$)
\State let set $\mathcal{S} = \{s_0\}$ 
\If {the job that departed from $s_0$ was the $k\supth$ job served from its batch}
\For {every server that is also serving jobs from this batch}
\State remove the job and add this server to set $\mathcal{S}$
\EndFor
\EndIf
\For {each $s \in \mathcal{S}$}
\If {there exists at least one batch in the buffer such that no job of this batch has been served by $s$}
\State among all such batches, find the batch that had arrived earliest
\State assign a job from this batch to $s$
\EndIf
\EndFor
\end{algorithmic}
\caption{First-come, first-served scheduling policy with redundant requests}
\label{alg:redundant_requests}
\end{algorithm*}

To the best of our knowledge, there has been little theoretical characterization of analysing under what settings sending redundant requests would help (and under what settings it would not help). Two exceptions relating to theoretical results in this area that we are aware of are~\cite{joshi2012coding} and~\cite{liang2013fast}. In~\cite{joshi2012coding}, Joshi et al. consider the arrival process to be Poisson, and the service to be i.i.d. memoryless, and provide bounds on the average latency faced by a batch in the steady state when the requests are sent (redundantly) to \textit{all} the servers. However, no comparisons are made with other schemes involving redundant requests, including the scheme of having no redundancy in the requests.
In fact, our work can be considered as complementary to that of~\cite{joshi2012coding}, in that we complete this picture by establishing that under the models considered therein, sending (redundant) requests to all servers is indeed the optimal choice. In~\cite{liang2013fast}, Liang and Kozat provide an approximate analysis of a system similar to that described in this paper under the assumption that arrivals follow a Poisson process; using insights from their approximations, they experiment with certain scheduling policies on the Amazon EC2 cloud. However, no analysis or metrics for accuracy of these approximations are provided, nor is there any treatment of whether these approximations lead to any useful upper or lower bounds.

\section{System Model: Centralized Buffer}\label{sec:model}
We will first describe the system model followed by an illustrative example. The model is associated to three parameters: two parameters $n$ and $k$ that are associated to the system, and the third parameter $\request$ that is associated to the redundant-requesting policy. The system comprises a set of $n$ servers. A request can be served by any \textit{arbitrary} $k$ \textit{distinct} servers out of this collection of $n$ servers. Several applications fall under the special case of $k=1$: a compute-cluster where computational tasks can be performed at any one of multiple processors, or a data-transmission scenario where requests comprise packets that can be transmitted across any one of multiple routes, or a distributed storage system with data replicated in multiple servers. Examples of settings with $k>1$ include: a distributed storage system employing an $(n,\ k)$ Reed-Solomon code wherein the request for any data can be served by downloading the data from any $k$ of the $n$ servers, or a compute-cluster where each job is executed at multiple processors in order to guard from possible errors during computation. 

The policy of redundant requesting is associated to a parameter $\request~(k\leq \request \leq n)$ which we call the `request-degree'. Each request is sent to $\request$ of the servers, and upon completion of any $k$ of these, it is deemed complete. To capture this, we consider each request as a \textit{batch} of $\request$ \textit{jobs}, wherein each of the $\request$ jobs can be served by any arbitrary $\request$ distinct servers. The batch is deemed served when any $k$ of its $r$ jobs are serviced. At this point in time, the remaining $(\request-k)$ jobs of this batch are removed from the system. Such a premature removal of a job from a server may lead to certain overheads: the server may need to remain idle for some (random) amount of time before it becomes ready to serve another job. We shall term this idle time as the \textit{removal cost}.
%We assume that the overheads associated to such a removal are negligible:  upon removal of a job that was being served by a server, the server immediately becomes available to serve any other job.
\begin{figure*}[t!]
\vspace{-1.3cm}
\centering
\subfloat[]{
\includegraphics[width=.18\textwidth]{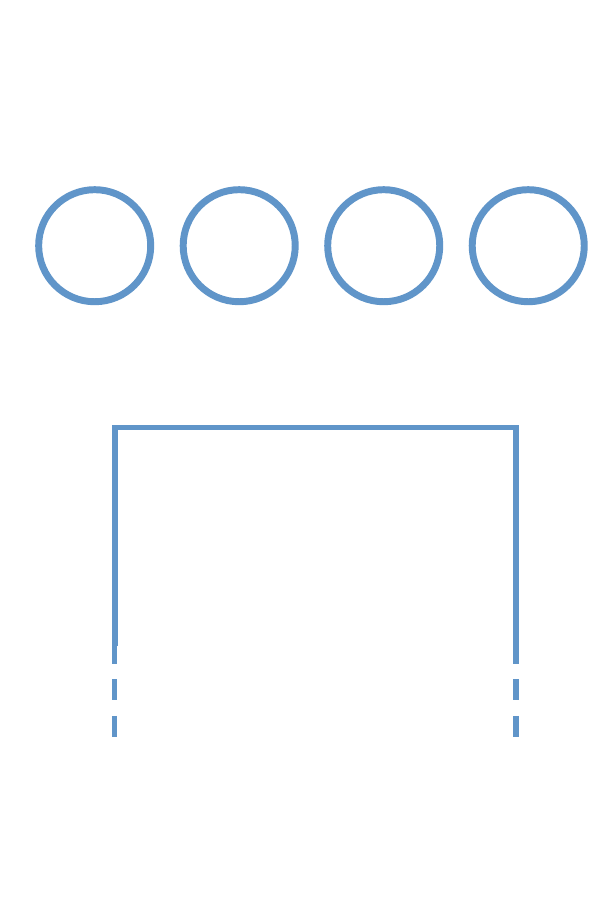}
\label{fig:scheduling_a}
}\hspace{\spaceBetweenScheduling}
\subfloat[]{
\includegraphics[width=.18\textwidth]{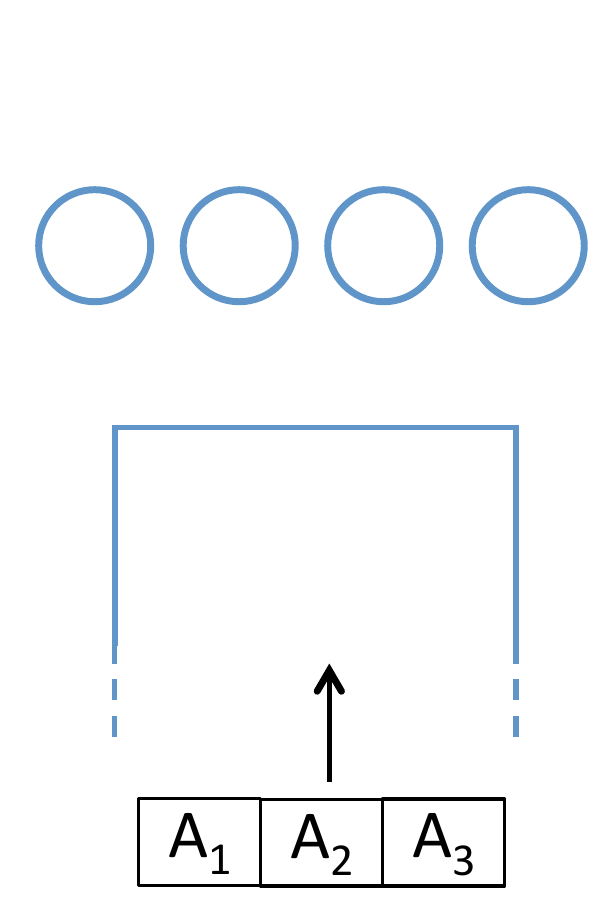}
\label{fig:scheduling_b}
}\hspace{\spaceBetweenScheduling}
\subfloat[]{
\includegraphics[width=.18\textwidth]{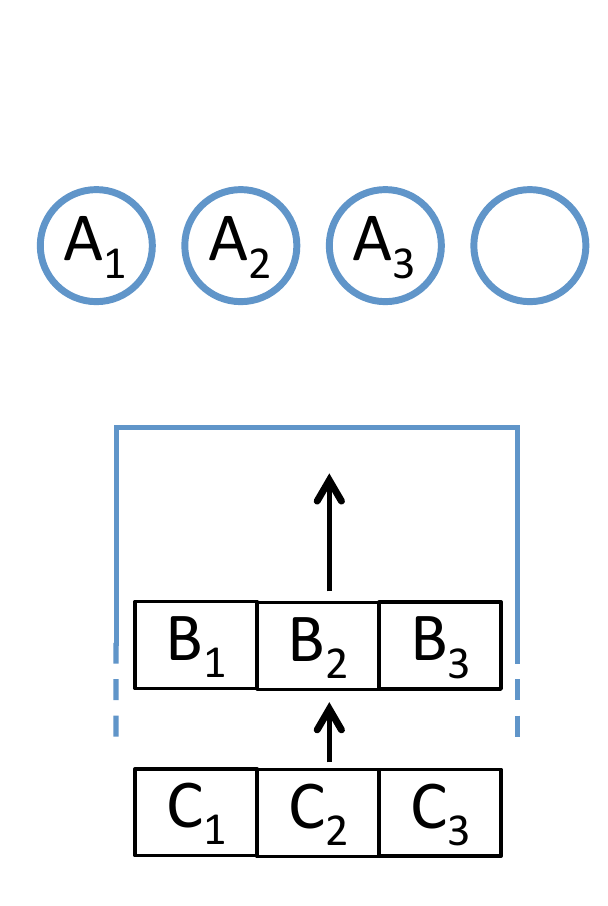}
\label{fig:scheduling_c}
}\hspace{\spaceBetweenScheduling}
\subfloat[]{
\includegraphics[width=.18\textwidth]{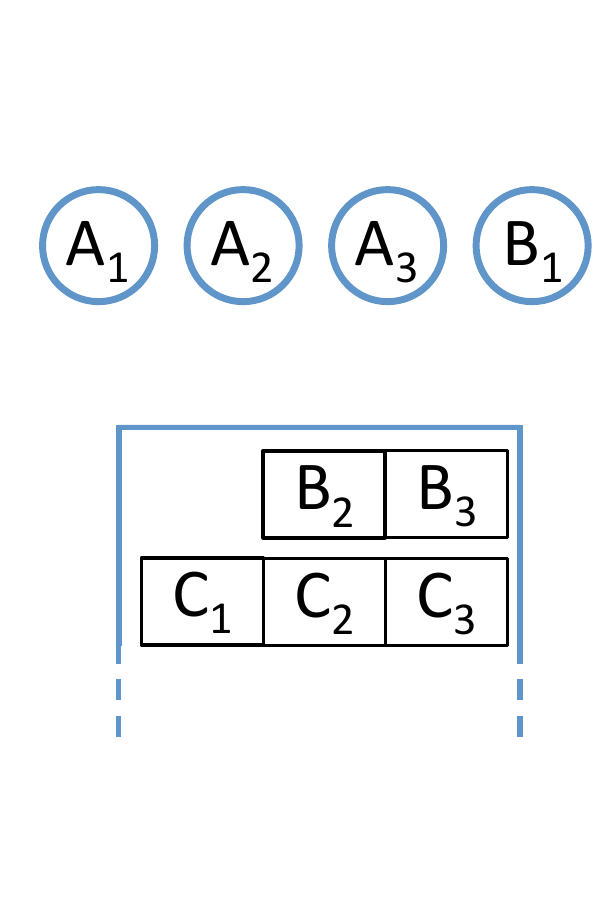}
\label{fig:scheduling_d}
}\\
\subfloat[]{
\includegraphics[width=.18\textwidth]{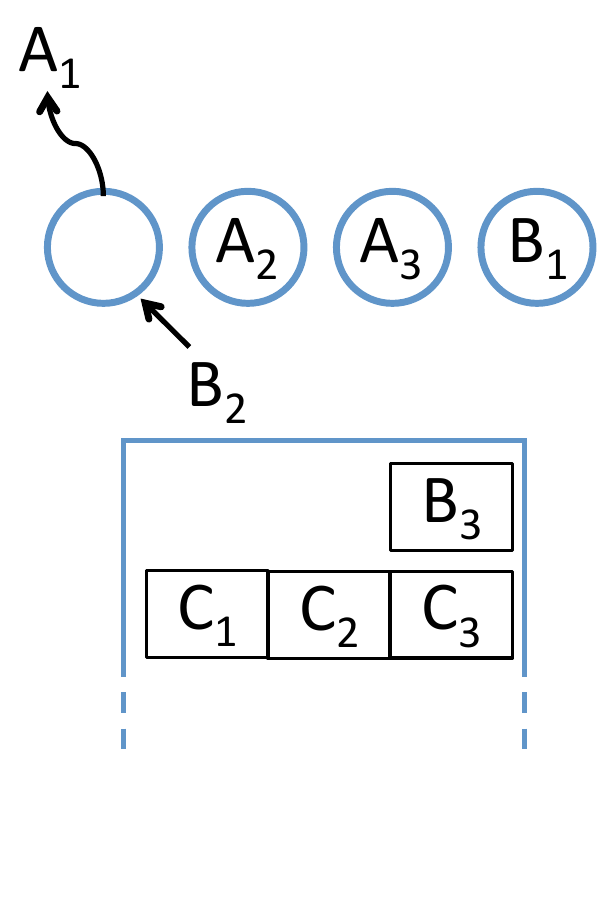}
\label{fig:scheduling_e}
}\hspace{\spaceBetweenScheduling}
\subfloat[]{
\includegraphics[width=.18\textwidth]{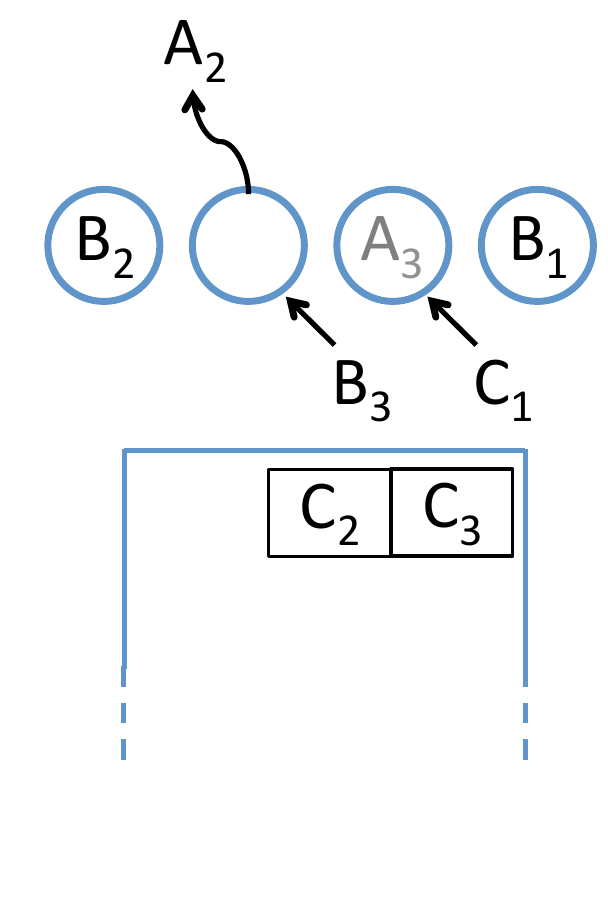}
\label{fig:scheduling_f}
}\hspace{\spaceBetweenScheduling}
\subfloat[]{
\includegraphics[width=.18\textwidth]{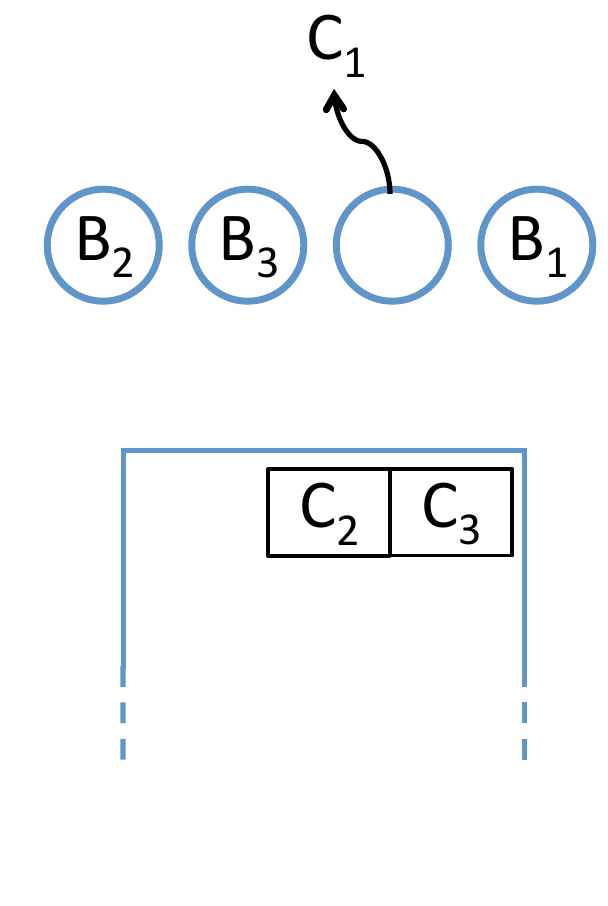}
\label{fig:scheduling_g}
}\hspace{\spaceBetweenScheduling}
\subfloat[]{
\includegraphics[width=.18\textwidth]{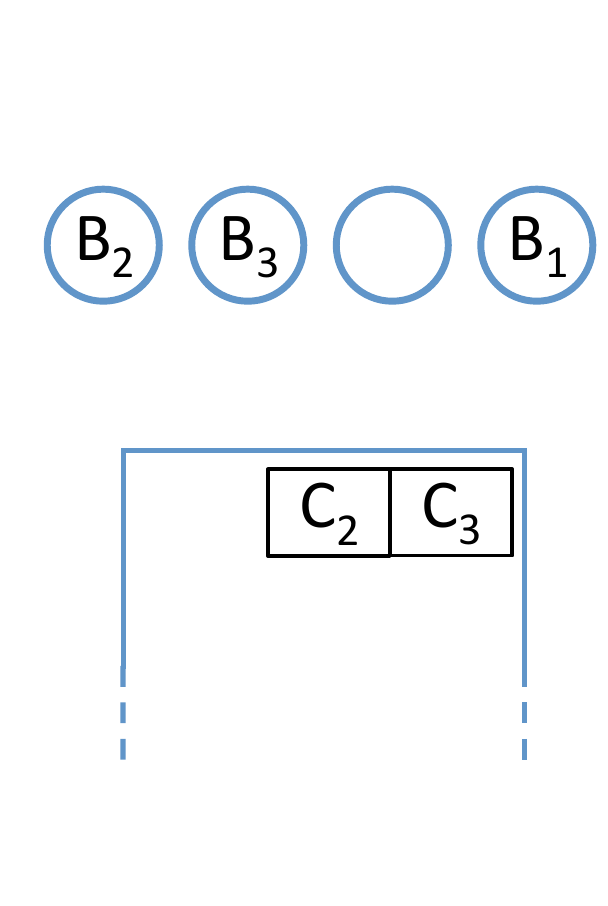}
\label{fig:scheduling_h}
}\caption{Illustration of the setting for parameters $n=4$, $k=2$ and request-degree $\request=3$, as described in Example~\ref{ex:setting}. %Requests take the form of batches, $A$, $B$, $C$ etc., each comprising of $\request=3$ jobs $\{A_1,A_2,A_3\}$, $\{B_1,B_2,B_3\}$, $\{C_1,C_2,C_3\}$, etc. out of which any $k=2$ must be served. Jobs of a batch must be served by distinct servers.
}
\label{fig:scheduling}
\end{figure*}

We assume that the time that a server takes to service a job is independent of the arrival and service times of other jobs. We further assume that the jobs are processed in a first-come-first-served fashion, i.e., among all the waiting jobs that an idle server can serve, it serves the one which had arrived the earliest. Finally, to be able to perform valid comparisons, we assume that the system is stable in the absence of any redundancy in the requests (i.e., when $\request=k$). The arrival process may be \textit{arbitrary}, and the only assumption we make is that the arrival process is independent of the present and past states of the system.

We consider a centralized queueing system in this section, where requests enter into a (common) buffer of infinite capacity. The choice of the server that serves a job may be made at any point in time. (This is in contrast to the distributed system considered subsequently in Section~\ref{sec:model_distributed}, wherein this choice must be made upon arrival of the request into the system). 

The scheduling algorithm is formalized in Algorithm~\ref{alg:redundant_requests}. Note that the case of $\request=k$ corresponds to the case where no redundancy is introduced in the requests, while $\request=n$ corresponds to maximum redundancy with each batch being sent to all the servers.

~\\
The following example illustrates the working of the system.
\begin{example}\label{ex:setting}
{\it
Fig.~\ref{fig:scheduling} illustrates the system model and the working of Algorithm~\ref{alg:redundant_requests} when $n=4$, $k=2$ and $\request=3$. The system has $n=4$ servers and a common buffer as shown in Fig.~\ref{fig:scheduling_a}. Let us denote the four servers (from left to right) as servers $1$, $2$, $3$ and $4$. Each request comes as a batch of $\request=3$ jobs, and hence we denote each batch (e.g., $A$, $B$, $C$, etc.) as a triplet of jobs (e.g., $\{A_1,A_2,A_3\}$, $\{B_1,B_2,B_3\}$, $\{C_1,C_2,C_3\}$, etc.). A batch is considered served if any $k=2$ of the $\request=3$ jobs of that batch are served.

Fig.~\ref{fig:scheduling_b} depicts the arrival of batch $A$. As shown in Fig.~\ref{fig:scheduling_c}, three of the idle servers begin serving the three jobs $\{A_1,A_2,A_3\}$. Fig.~\ref{fig:scheduling_c} depicts the arrival of batch $B$ followed by batch $C$. Server $4$ begins service of job $B_1$ as shown in Fig.~\ref{fig:scheduling_d}, while the other jobs wait in the buffer. Now suppose server $1$ completes servicing job $A_1$ (Fig.~\ref{fig:scheduling_e}). This server now becomes idle to serve any of the jobs remaining in the buffer. We allow jobs to be processed in a first-come first-served manner, and hence server $1$ begins servicing job $B_2$ (assignment of $B_3$ instead would also have been valid). Next, suppose the second server completes service of $A_2$ before any other servers complete their current tasks (Fig.~\ref{fig:scheduling_f}). This results in the completion of a total of $k=2$ jobs of batch $A$, and hence batch $A$ is deemed served and is removed the system. In particular, job $A_3$ is removed from server $3$ (this may cause the server to remain idle for some time, depending on the associated removal cost). Servers $2$ and $3$ are now free to serve other jobs in the buffer. These are now populated with jobs $B_3$ and $C_1$ respectively. Next suppose server $3$ completes serving $C_1$ (Fig.~\ref{fig:scheduling_g}). In this case, since server $3$ has already served a job from batch $C$, it is not allowed to service $C_2$ or $C_3$ (since the jobs of a batch must be processed by distinct servers). Since there are no other batches waiting in the buffer, server $3$ thus remains idle (Fig.~\ref{fig:scheduling_h}).
}
\end{example}
%~\\

%The question we aim to answer in our work, in a nutshell, is: \textit{For what service time distributions and arrival processes do redundant requests necessarily help? If redundant requests help only when the arrival rate is low, what is the threshold on the arrival rate before which they stop helping?}

\section{Analytical Results for the Centralized Buffer Setting}\label{sec:analytical}
In this section, we consider the model presented in Section~\ref{sec:model} that has a centralized buffer. We find redundant-requesting policies that minimize the average latency under various settings. This minimization is not only over redundant-requesting policies with a fixed value of the request-degree $r$ (as described in Section~\ref{sec:model}) but also over policies that can choose different request-degrees for different batches. The proofs of these results are provided in Appendix~\ref{app:proofs}.

 %service times of the jobs being i.i.d. with an exponential distribution. 
%We first show that when each request is required to be processed by any \textit{one} of the $n$ servers, and the service times follow an exponential distribution, increasing the redundancy in the requests monotonically reduces the latency for all system loads.
%Theorems~\ref{thm:rep_flood},~\ref{thm:memoryless_general_k},~\ref{thm:heavier_flood} and~\ref{thm:lighter_noflood} below present analytical results for this setting.
The first two results, Theorems~\ref{thm:rep_flood} and~\ref{thm:memoryless_general_k}, consider the service times to follow an exponential (memoryless) distribution. 

\begin{theorem}[memoryless service, no removal cost, $k=1$]\label{thm:rep_flood}
Consider a system with $n$ servers such that any one server suffices to serve any request, the service-time is i.i.d. memoryless, and jobs can be removed instantly from the system. For any $\request_1 < \request_2$, the average latency in a system with request-degree $\request_1$ is larger than the average latency in a system with  request-degree $\request_2$. Furthermore, the distribution of the buffer occupancy in the system with  request-degree $\request_1$ dominates (is larger than) that of the system with  request-degree $\request_2$. Finally, among all possible redundant requesting policies, the average latency is minimized when each batch is sent to all $n$ servers, i.e., when request-degree $\request=n$.
\end{theorem}

~\\
\begin{theorem}[memoryless service, no removal cost, general $k$]\label{thm:memoryless_general_k}
Consider a system with $n$ servers such that any $k$ of them can serve a request, the service-time is i.i.d. memoryless, and jobs can be removed instantly from the system. The average latency is minimized when all batches are sent to all the servers, i.e., when $\request=n$ for every batch. Furthermore, the distribution of the buffer occupancy in the system with request-degree $\request=n$ is strictly dominated by (i.e., is smaller than) a system with any other request-degree. 
\end{theorem}
Fig.~\ref{fig:penalty_zero} depicts simulations that corroborate this result. 

\begin{figure*}[t!]
\centering
\begin{minipage}{.46\textwidth}
\centering
\includegraphics[width=.9\textwidth]{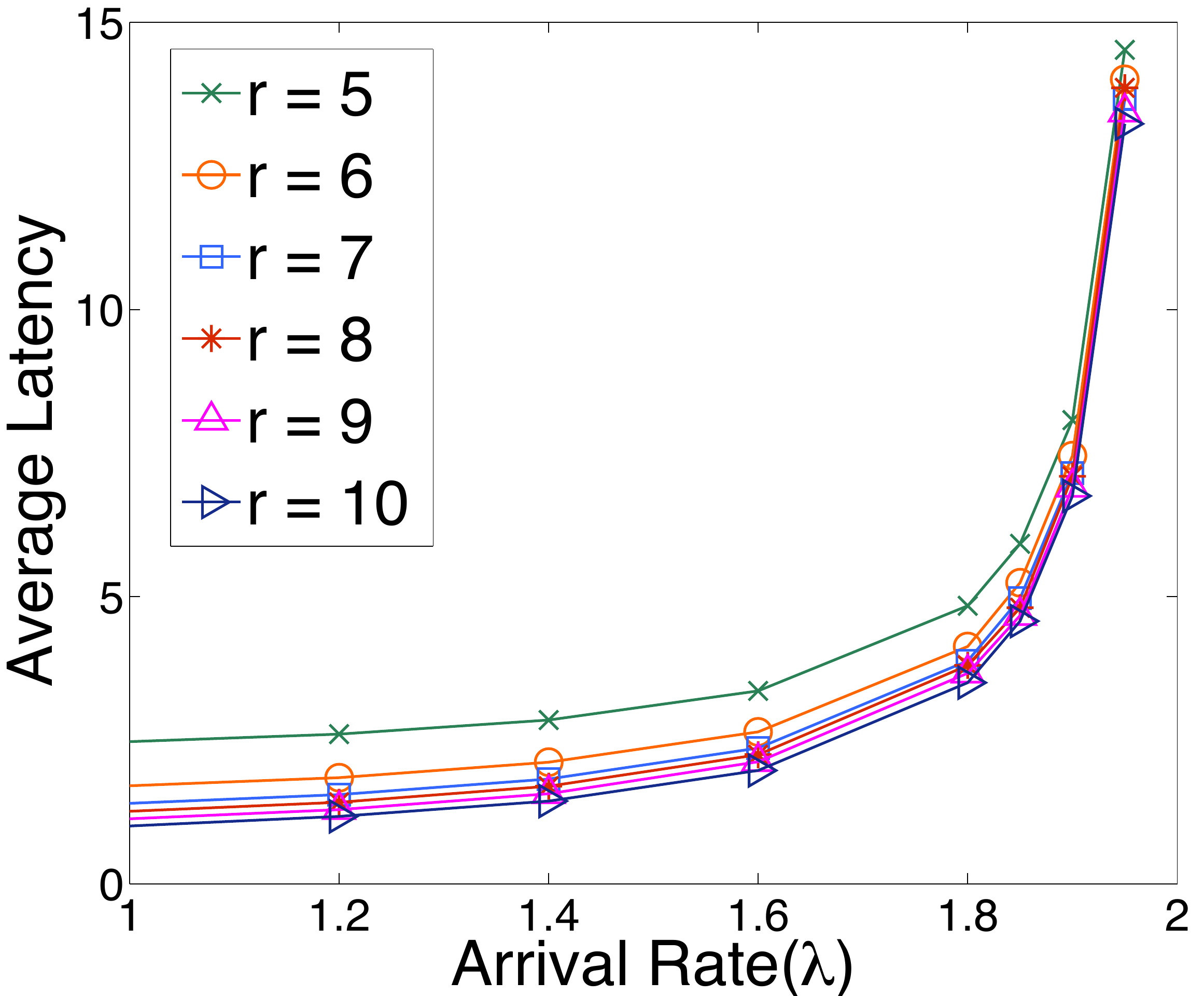}
\caption{Simulation results showing the  average latency for various values of the request-degree $\request$, in a $(n=10,\ k=5)$ system with i.i.d. memoryless service with rate $1$ and arrivals following a Poisson process with rate $\lambda$.}
\label{fig:penalty_zero}
\end{minipage}
~~
\begin{minipage}{.46\textwidth}
\centering
\includegraphics[width=.9\textwidth]{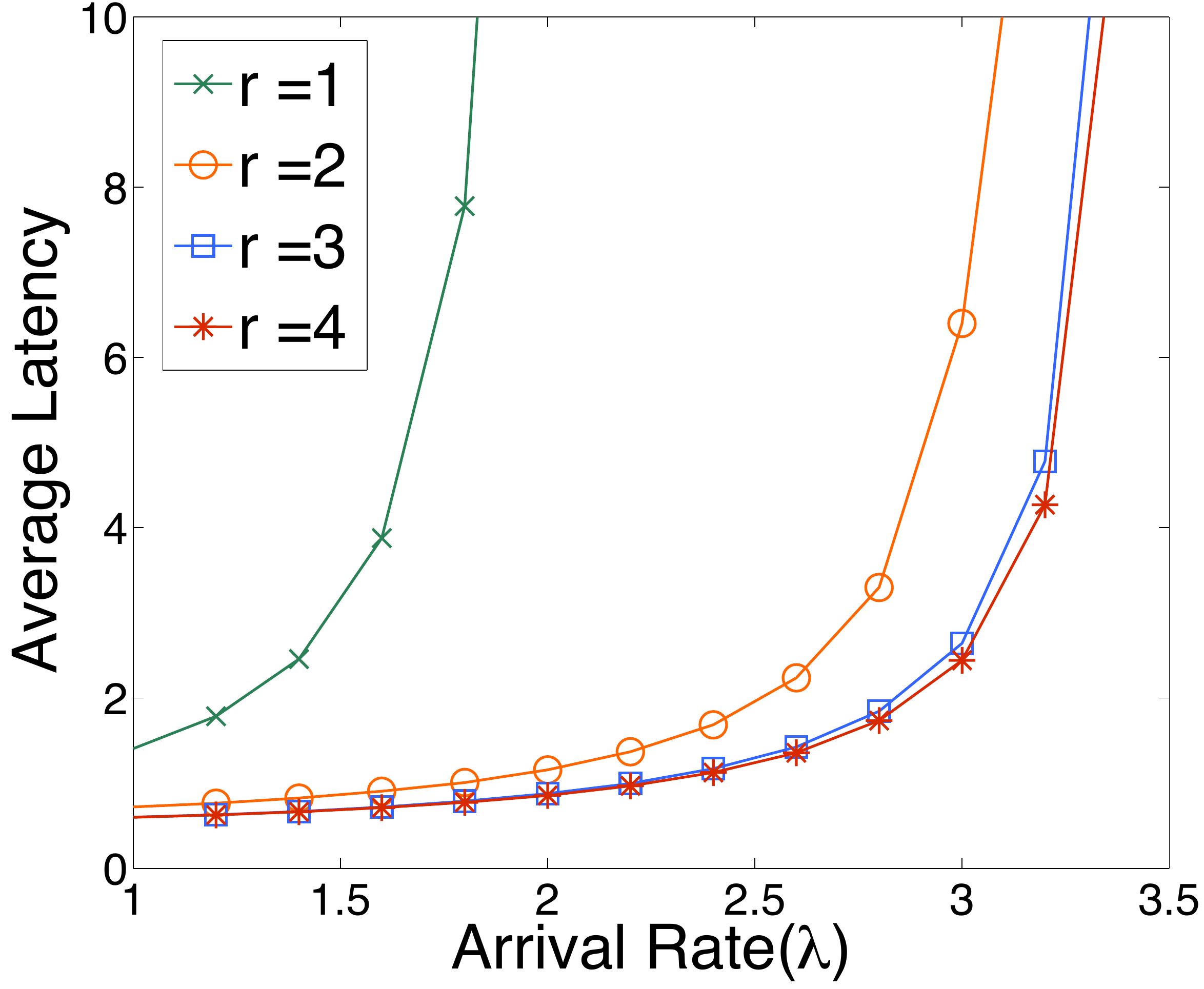}
\caption{Average latency in a $(n=4,\ k=1)$ system with a heavy-everywhere service time. The service time is distributed as a mixture of exponential distributions and the arrival process is Poisson with a rate $\lambda$.}
\label{fig:flood_hyper_exp}
\end{minipage}
\end{figure*}

\begin{figure*}[t!]
\medmuskip=0\medmuskip
\thinmuskip=0\thinmuskip
\thickmuskip=0\thickmuskip
\begin{minipage}{.483\textwidth}
\centering
\includegraphics[width=.9\textwidth]{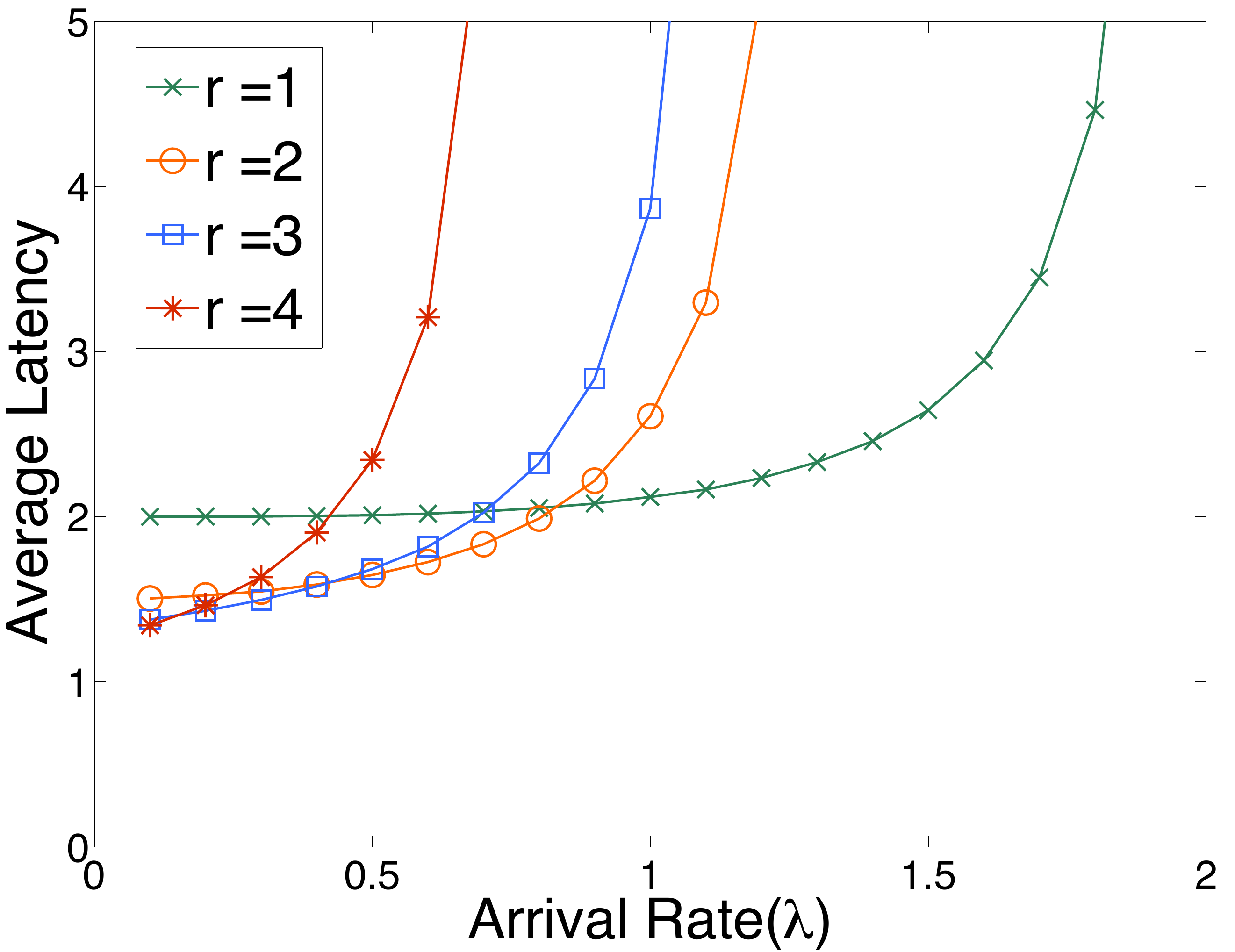}
\caption{Average latency in a $(n=4,\ k=1)$ system with a light-everywhere service time. The service time is distributed as an exponential distribution with rate $1$ shifted by a constant of value $1$ and the arrival process is Poisson with a rate $\lambda$.}
\label{fig:flood_const_plus_exp}
\end{minipage}
~~~
\begin{minipage}{.483\textwidth}
\centering
\includegraphics[width=.9\textwidth]{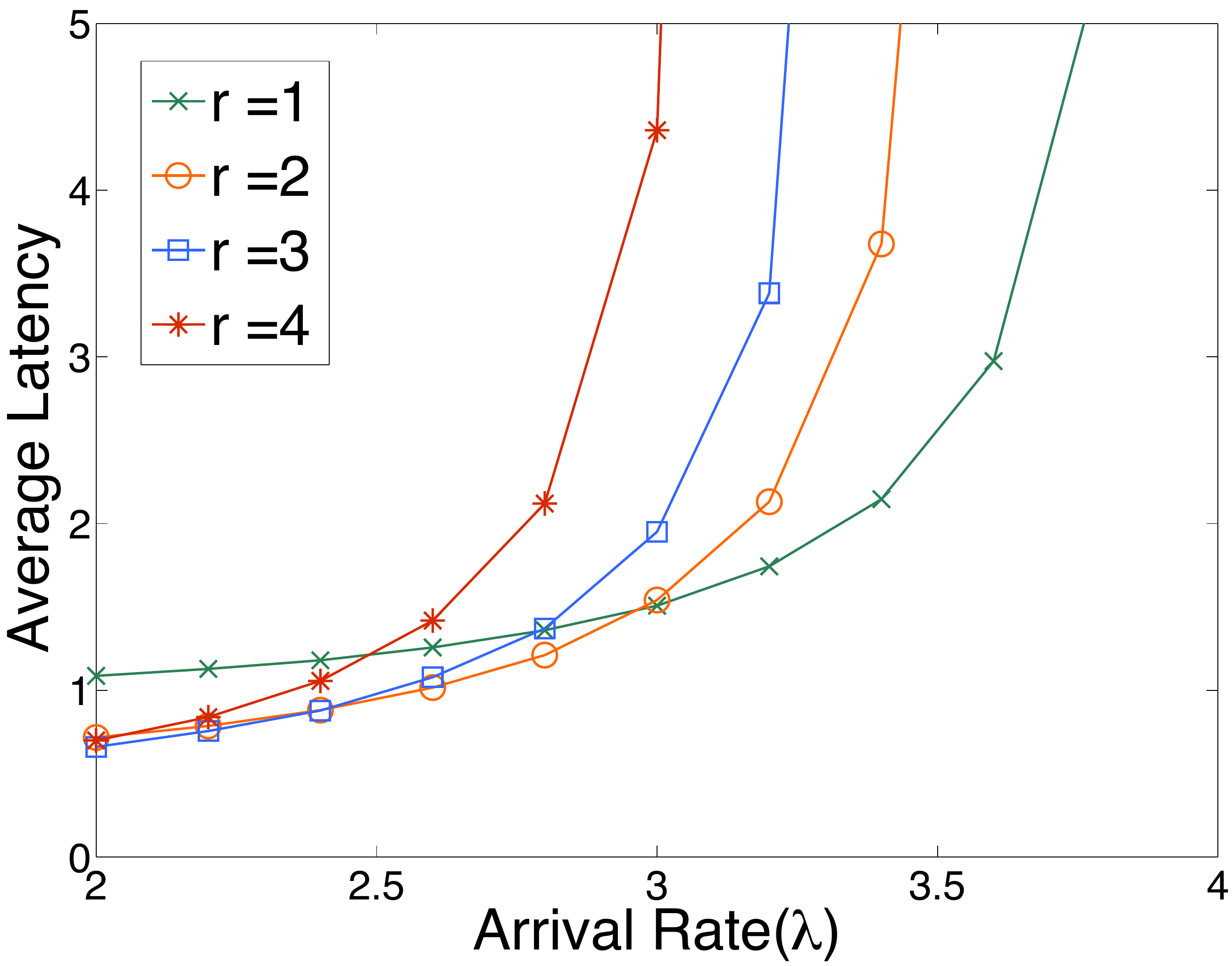}
\caption{Average latency in a $(n=4,\ k=1)$ system with the service time following an exponential distribution with rate $1$, wherein removing an unfinished job from a server requires the server to remain idle for a time distributed exponentially with rate $10$.}
\label{fig:penalty_10percent}
\end{minipage}
\end{figure*}

We now move on to some more general classes of service-time distributions. The first class of distributions is what we term \textit{heavy-everywhere}, defined as follows.
\begin{definition}[\textbf{Heavy-everywhere distribution}]
A distribution on the non-negative real numbers is termed heavy-everywhere if for every pair of values $a > 0$ and $b \geq 0$ with $P(X>b)>0$, the distribution satisfies 
\beq 
P(X>a+b\ |\ X>b) \  \geq \  P(X>a)~. \label{eq:heavier_definition}
\eeq 
\end{definition}
In words, under a heavy-everywhere distribution, the need to wait for a while makes it more likely that a bad event has occurred, thus increasing the possibility of a greater wait than usual.

For example, a mixture of independent exponential distributions satisfies~\eqref{eq:heavier_definition} and hence is heavy-everywhere. Some properties of heavy-everywhere distributions are discusses in Appendix~\ref{app:everywhere_dist}.

A second class of distributions is what we call \textit{light-everywhere} distributions, defined as follows.
\begin{definition}[\textbf{Light-everywhere distribution}]
A distribution on the non-negative real numbers is termed light-everywhere if for every pair of values $a > 0$ and $b \geq 0$ with $P(X>b)>0$, the distribution satisfies 
\beq 
P(X>a+b\ |\ X>b) \  \leq \  P(X>a)~. \label{eq:lighter_definition}
\eeq 
\end{definition}
In words, under a light-everywhere distribution, waiting for some time brings you closer to completion, resulting in a smaller additional waiting time.

For example, an exponential distribution that is shifted by a positive constant is light-everywhere, and so is the uniform distribution.  Some properties of light-everywhere distributions are discussed in Appendix~\ref{app:everywhere_dist}

%Note that an exponential distribution is both heavy-everywhere and light.
\begin{figure*}[t!]
\centering
\subfloat[]{
\includegraphics[width=.18\textwidth]{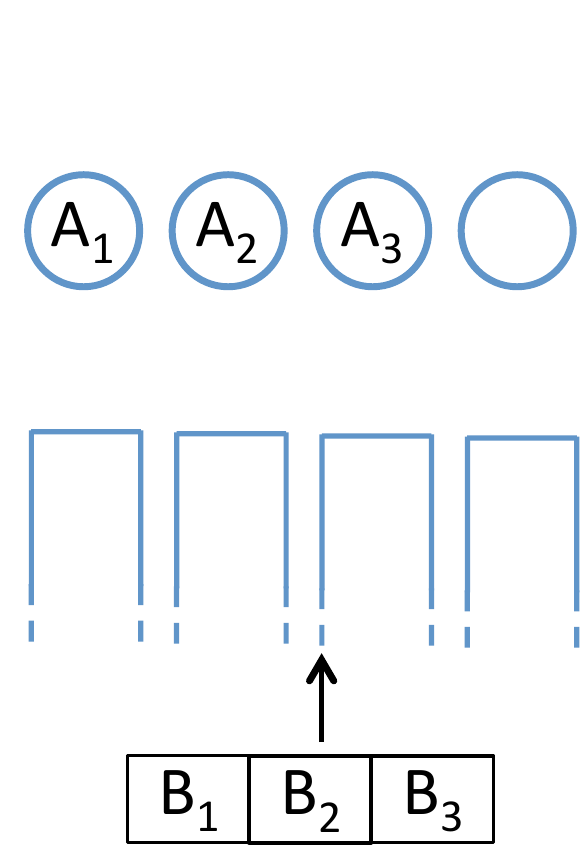}
\label{fig:distributed_scheduling_a}
}\hspace{\spaceBetweenScheduling}
\subfloat[]{
\includegraphics[width=.18\textwidth]{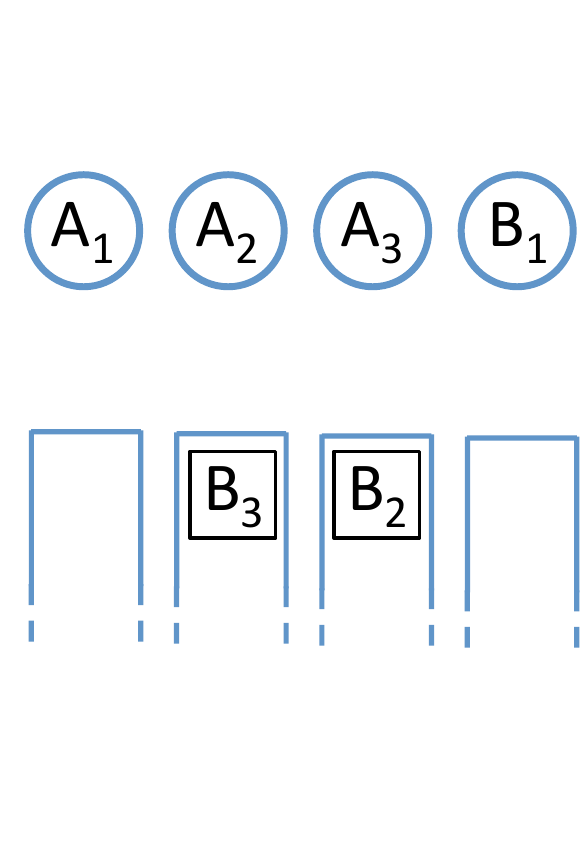}
\label{fig:distributed_scheduling_b}
}\hspace{\spaceBetweenScheduling}
\subfloat[]{
\includegraphics[width=.18\textwidth]{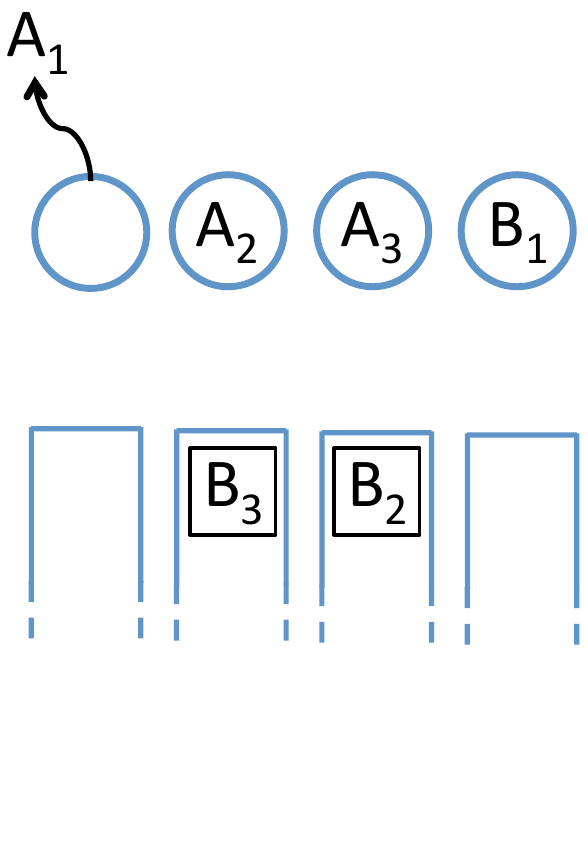}
\label{fig:distributed_scheduling_c}
}\hspace{\spaceBetweenScheduling}
\subfloat[]{
\includegraphics[width=.18\textwidth]{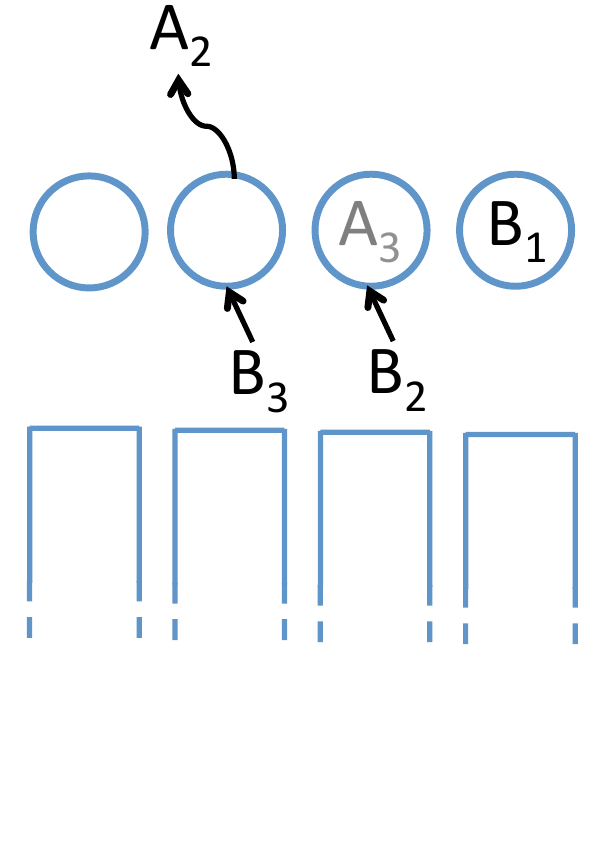}
\label{fig:distributed_scheduling_d}
}\caption{Illustration of the setting with distributed buffers for parameters $n=4$, $k=2$ and request-degree $\request=3$, as described in Example~\ref{ex:setting_distributed}.}
\label{fig:distributed_scheduling}
\end{figure*}
The following theorems present results for systems with service-times belonging to one of these two classes of distributions. 

\begin{theorem}[heavy-everywhere service, no removal cost, $k=1$, high load]\label{thm:heavier_flood}
Consider a system with $n$ servers such that any one server suffices to serve any request, the service-time is i.i.d. heavy-everywhere, and jobs can be removed instantly from the system. When the system has a 100\% server utilization, the average latency is minimized when each batch is sent to all $n$ servers, i.e., when $r=n$ for each batch.
\end{theorem}
This is corroborated in Fig.~\ref{fig:flood_hyper_exp} which depicts simulations with the service time $X$ distributed as a mixture of exponentials:
%\[ P(X>x) = 0.2 e^{-0.1x}+0.8e^{-x}\]
\[ X \sim \left\{ \begin{tabular}{cc} exp(\,rate\,=\,0.1) & w.p. 0.2\\ exp(\,rate\,=\,1) & w.p. 0.8~. \end{tabular}\right.~\]
Note that Theorem~\ref{thm:heavier_flood} addresses only the scenario of high loads and predicts minimization of latency when $\request=n$ in this regime; simulations of Fig.~\ref{fig:flood_hyper_exp} further seem to suggest that the policy of $\request=n$ minimizes the average latency for all loads. Similar phenomena are observed in simulations for $k>1$.

\begin{theorem}[light-everywhere service, any removal cost, $k=1$, high load]\label{thm:lighter_noflood}	
Consider a system with $n$ servers such that any one server suffices to serve any request, and the service-time is i.i.d. light-everywhere. When the system has a 100\% server utilization, the average latency is minimized when there is no redundancy in the requests, i.e., when $\request=k~(=1)$ for all batches.
\end{theorem}

This is corroborated in Fig.~\ref{fig:flood_const_plus_exp} which depicts simulations with the service time $X$ distributed as a sum of a constant and a value drawn from an exponential distribution:
\[
P(X>x) = \begin{cases}e^{-(x-1)} & \textrm{if } x \geq 1\\
1 &\textrm{otherwise}~.
\end{cases}
\]
We observe in Fig.~\ref{fig:flood_const_plus_exp} that at high loads, the absence of any redundant requests (i.e., $r=1$) minimizes the average latency, which is as predicted by the theory. We also observe in the simulations for this setting that redundant requests do help when arrival rates are low, but start hurting beyond a certain threshold on the arrival rate. Similar phenomena are observed in simulations for $k>1$.

The next theorem revisits memoryless service times, but under non-negligible removal costs.
\begin{theorem}[memoryless service, non-zero removal cost, $k=1$, high load]\label{thm:removalCosts_memless}	
Consider a system with $n$ servers such that any one suffices server to serve any request, and the service-time is i.i.d. memoryless, and removal of a job from a server incurs a non-zero delay. When the system has a 100\% server utilization, the average latency is minimized when there is no redundancy in the requests, i.e., when $\request=k~(=1)$ for all batches.
\end{theorem}

Fig.~\ref{fig:penalty_10percent} presents simulation results for such a setting. The figure shows that under this setting, redundant requests lead to a higher latency at high loads, as predicted by theory.

\section{System Model: Distributed Buffers}\label{sec:model_distributed}
The model with distributed buffers closely resembles the case of a centralized buffer. The only difference is that in this distributed setting, each server has a buffer of its own, and the jobs of a batch must be sent to some $\request$ of the $n$ buffers \textit{as soon as the batch arrives in the system}. The protocol for choosing these $\request$ servers for each batch may be arbitrary for the purposes of this paper, but for concreteness, the reader may assume that the $\request$ least-loaded buffers are chosen. The setting with distributed buffers is illustrated in the following example.
~\\
\begin{example}\label{ex:setting_distributed}
{\it
Fig.~\ref{fig:distributed_scheduling} illustrates the system model and the working of the system in the distributed setting, for parameters $n=4$, $k=2$ and $\request=3$. The system has $n=4$ servers, and each of these servers has its own buffer, as shown in Fig.~\ref{fig:distributed_scheduling_a}. Denote the four servers (from left to right) as servers $1$, $2$, $3$ and $4$. Fig.~\ref{fig:distributed_scheduling_a} depicts a scenario wherein batch $A$ is already being served by the first three servers, and batch $B$ just arrives. The three servers (buffers) to which batch $B$ will be sent to must be selected at this time. Suppose the algorithm chooses to send the batch to buffers $2$, $3$ and $4$ (Fig.~\ref{fig:distributed_scheduling_b}). Now suppose server $1$ completes service of job $A_1$ (Fig.~\ref{fig:distributed_scheduling_c}). Since there is no job waiting in the first buffer, server $1$ remains idle. Note that in contrast, a centralized setting would have allowed the first server to start processing either job $B_2$ or $B_3$. Next, suppose server $2$ completes service of job $A_2$ (Fig.~\ref{fig:distributed_scheduling_d}). With this, $k=2$ jobs of batch $A$ are served, and the third job $A_3$ is thus removed. Servers $2$ and $3$ can now start serving jobs $B_3$ and $B_2$ respectively.
}
\end{example}

\section{Analytical Results for the Distributed Buffers Setting}\label{sec:analytical_distributed}
As in the centralized setting of Section~\ref{sec:analytical}, we continue to assume that the service-time distributions of jobs are i.i.d. and the system operates on a first-come-first-served basis. The following theorems prove results that are distributed counterparts of the results of Section~\ref{sec:analytical}.

\begin{theorem}[memoryless service, no removal cost, general $k$]\label{thm:memoryless_general_k_distributed}
Consider a system with $n$ servers such that any $k$ of them can serve a request, the service-time is i.i.d. memoryless, and jobs can be removed instantly from the system. The average latency is minimized when all batches are sent to all the servers, i.e., when $\request=n$ for every batch.
\end{theorem}

\begin{theorem}[heavy-everywhere service, no removal cost, $k=1$, high load]\label{thm:heavy_distributed}
Consider a system with $n$ servers such that any one server suffices to serve any request, the service-time is i.i.d. heavy-everywhere, and jobs can be removed instantly from the system. When the system has a 100\% server utilization, the average latency is minimized when each batch is sent to all $n$ servers, i.e., when $r=n$ for each batch.
\end{theorem}

\begin{theorem}[light-everywhere service, any removal cost, $k=1$, high load]\label{thm:lighter_noflood_distributed}	
Consider a system with $n$ servers such that any one server suffices to serve any request, and the service-time is i.i.d. light-everywhere. When the system has a 100\% server utilization, the average latency is minimized when there is no redundancy in the requests, i.e., when $\request=k~(=1)$ for all batches.
\end{theorem}

\begin{theorem}[memoryless service, non-zero removal cost, $k=1$, high load]\label{thm:removalCosts_memless_distributed}	
Consider a system with $n$ servers such that any one suffices server to serve any request, and the service-time is i.i.d. memoryless, and removal of a job from a server incurs a non-zero delay. When the system has a 100\% server utilization, the average latency is minimized when there is no redundancy in the requests, i.e., when $\request=k~(=1)$ for all batches.
\end{theorem}

%\section{General Proof Technique}\label{sec:proof_technique}
\section{Conclusions and Open Problems}\label{sec:conclusion}
The prospect of reducing latency by means of redundant requests has garnered significant attention among practitioners in the recent past (e.g.,~\cite{snoeren2001mesh,andersen2005improving,pitkanen2007redundancy,han2011rpt,ananthanarayanan2012let,huang2012erasure,vulimiri2012more,dean2013tail,liang2013fast,stewart2013zoolander,flach2013reducing}). Many recent works empirically evaluate the latency performance of redundant requests under diverse settings. The goal of our work is to analytically characterize the settings under which redundant requests help (and when they hurt), and to design scheduling policies that employ redundant-requesting to reduce latency.  In this paper, we propose a model that captures key features of such systems, and under this model we analytically characterize several settings wherein redundant requests help and where they don't. For each of these settings, we also derive the optimal redundant-requesting policy. 

While we have characterized \textit{when redundant requests help} for several scenarios in this paper, the characterization for many more general settings remains open. Some questions that immediately arise are:
\begin{itemize}[$\bullet$]
\item What is the optimal redundant-requesting policy for service-time distributions and removal-costs not considered in this paper ?
\item We observed in the simulations  (e.g., Fig.~\ref{fig:flood_const_plus_exp}) that for several service-time distributions, redundant requests start hurting when the system is loaded beyond a certain threshold. In the future, we wish to use the insights developed in this paper to analytically characterize this threshold.
\item What happens when the requests or the servers are heterogeneous, or if the service-times of different jobs of a batch are not i.i.d. ? 
\item What about other metrics such as the tails of the latency, or a quantification of the amount of gains achieved via redundant requests ?
\item If we allow choosing different values of the request-degree $\request$ adaptively for different batches, what is the minimal information about the state of the system required to make this choice? What are the optimal scheduling policies in that case ?
\item In certain settings, one may be constrained with each request having the ability to get served by only a specific $m~(<n)$ of the $n$ servers. It remains to investigate which of the results for $m=n$ carry over to this setting of $m<n$ (see, for example, Fig~\ref{fig:more_servers})?
\end{itemize}
\begin{figure*}
\centering
\includegraphics[width=.5\textwidth]{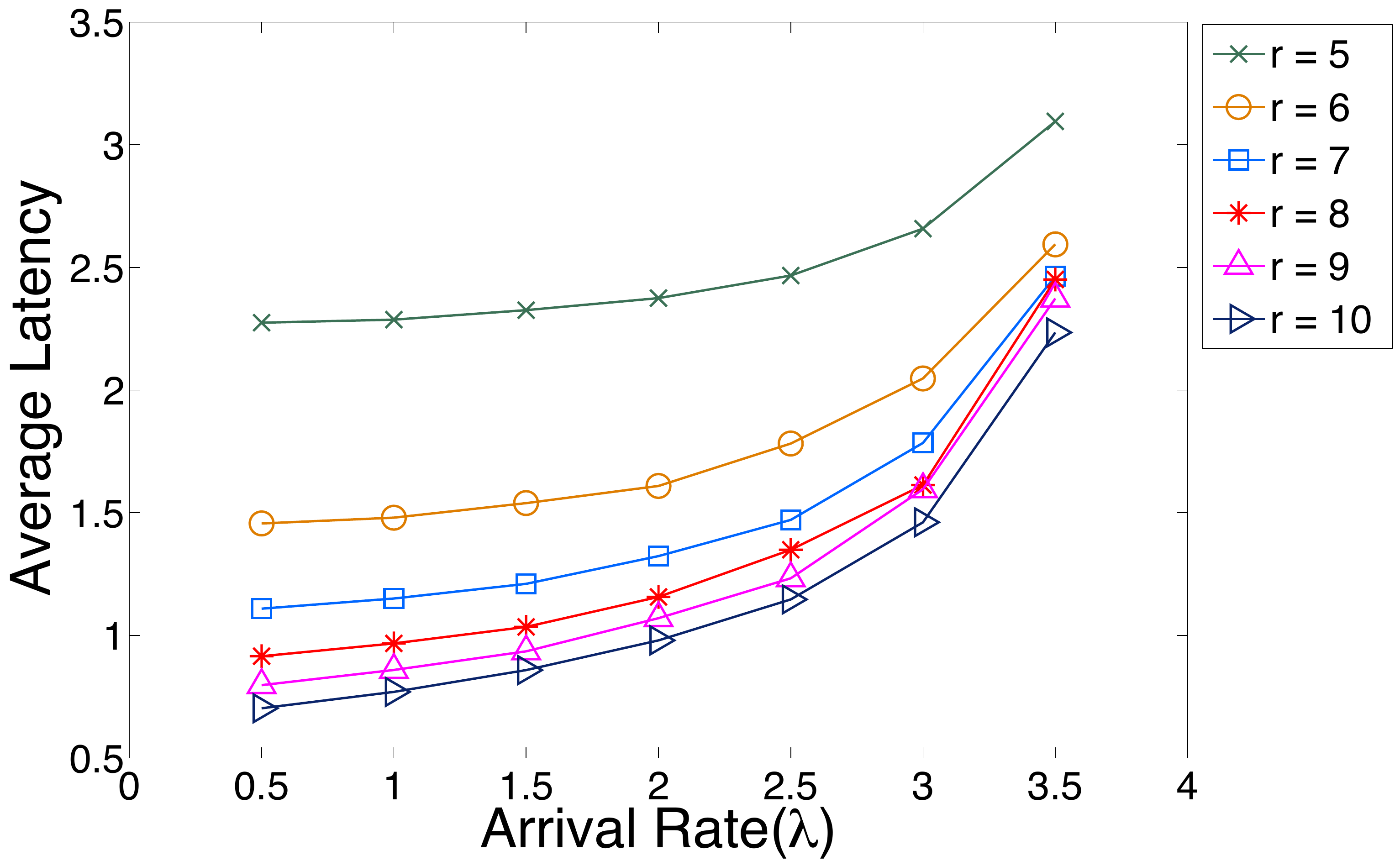}
\caption{Average latency in a system where the total number of servers is $n=20$ but for each request, a specific subset of only $m=10$ of these servers have the ability to serve it. The request must be handled by any $k=5$ distinct servers from this set of $m=10$ servers. For each batch, these $m=10$ servers is uniform from the set of $n=20$ for each batch, and is independent across batches. The service is distributed exponentially with rate $1$ and there is no removal cost. One can see that the average latency reduces with an increase in the redundancy in the requests.}
\label{fig:more_servers}
\end{figure*}

\bibliographystyle{IEEEtran}

\appendices
\section{Heavy-everywhere and light-everywhere distributions}\label{app:everywhere_dist}
In this section we derive some properties of heavy-everywhere~\eqref{eq:heavier_definition} and light-everywhere~\eqref{eq:lighter_definition} classes of distributions. We also provide examples of distributions that fall into these classes. We first state the results, following which we provide the proofs.

\begin{proposition}
The expected value of the minimum of $n$ random variables, each drawn independently from a distribution that is heavy-everywhere, is no larger than $\frac{1}{n}$ times the expected value of that distribution. The expected value of the minimum of $n$ random variables, each drawn independently from a distribution that is light-everywhere, is no smaller than $\frac{1}{n}$ times the expected value of that distribution.
\label{prop:heavy_light_mean}
\end{proposition}

\begin{proposition}
Consider a finite set of independent random variables $X_1,\ldots,X_L$, each of whose (marginal) distributions is heavy-everywhere, such that for every $i,j$ and every $a\geq 0,\ b\geq 0$,
\[P(X_i>a) > P(X_j>a) \Rightarrow P(X_i>b) \geq P(X_j>b)~.\] 
Then, any mixture of $X_1,\ldots,X_L$ is also a heavy-everywhere distribution.
\label{prop:heavy_mix_is_heavy}
\end{proposition}

\begin{proposition}\label{prop:heavy_examples}
The following distributions are heavy-everywhere:
\begin{enumerate}
\item A mixture of a finite number of independently drawn exponential distributions.
\item A Weibull distribution with scale parameter smaller than $1$, i.e., with a pdf 
\[ f(x) = \frac{k}{\lambda} {\left(\frac{x}{\lambda}\right)}^{k-1} e^{-{\left(x/\lambda\right)}^k} \]
for any $k \in (0,1]$ and any $\lambda>0$.
\end{enumerate}
\end{proposition}

%TODO: find more distributions that are heavy-everywhere. Log normal ? 

\begin{proposition}\label{prop:light_sum_is_light}
The sum of a finite number of independent random variables, each of which has a (marginal) distribution that is light-everywhere, also has a distribution that is light-everywhere.
\end{proposition}

\begin{proposition}\label{prop:light_examples}
The following distributions are light-everywhere:
\begin{enumerate}
\item For any $c>0$, the constant distribution with entire mass on $c$.
\item An exponential distribution that is shifted by a positive constant.
\item The uniform distribution.
\item For any pair of non-negative constants $c_1$ and $c_2$ with $2c_1 > c_2>c_1$, a distribution with its support comprising only the two constants $c_1$ and $c_2$.
\end{enumerate}
\end{proposition}

We now present the proofs of these claims.

%\begin{proposition}
%Consider any $a>0$ and $b>0$. Suppose $X$ is drawn from a distribution that is everywhere-heavy. Then
%\[ 
%P(X>a+b|X>b) > P(X>a) \Rightarrow P(X>a'+b|X>b) > P(X>a')~\forall~a' \geq a~.
%\]
%Suppose $Y$ is drawn from a distribution that is light-everywhere. Then
%\[ 
%P(Y>a+b|Y>b) < P(Y>a) \Rightarrow P(Y>a'+b|Y>b) < P(Y>a')~\forall~a' \geq a~.
%\]
%\end{proposition}
%\begin{proof}
%Consider $X$ to be drawn from an everywhere-heavy distribution. Since $a' \geq a$, we have
%\beq
%P(X>a'+b|X>b) &=& P(X>a'+b, X>a+b |X>b)\\
%&=& P(X>a'+b|X>a+b) P(X>a+b|X>b)\\
%&>& P(X>a'+b|X>a+b) P(X>a)\\
%&\geq & P(X>a'+b|X>a+b) P(X>a+b)\\
%&= & P(X>a'+b, X>a+b) \\
%\end{proof}
%\begin{proof}
%The distribution under consideration can be expressed as, for all $x>0$,
%\[
%P(X > x) = \min\{e^{-\mu(x-c)},1\}
%\]
%for some $c \geq 0$ and any $\mu > 0$.
%We need to show that~\eqref{eq:lighter_definition} holds for all $a>0$ and $b>0$.
%We divide the proof into three cases:\\
%Case I, $a\geq c,b \geq c$: In this case
%\bea
%P(X>a+b) & = & \min\{e^{-\mu(a+b-c)},1\}\\
%&=& e^{-\mu(a+b-c)}\\
%&=& e^{-\mu(a-c)} e^{-\mu(b-c)} e^{-\mu c}\\
%&\leq& e^{-\mu(a-c)} e^{-\mu(b-c)}\\
%&=& \min\{e^{-\mu(a-c)},1\} \min\{e^{-\mu(b-c)},1\}\\
%&=& P(X>a)P(X>b)~.
%\eea
%
%\noindent Case II: $a < c$, $b < c$: $P(X>a+b) \leq 1$ but $P(X>a)=1$ and $P(X>b)=1$.
%
%\noindent Case III: $a \leq c$, $b> c$: In this case, $P(X>a)=1$. Using this fact, we get
%\bea
%P(X>a+b) & = & \min\{e^{-\mu(a+b-c)},1\}\\
%&=& e^{-\mu(a+b-c)}\\
%&\leq& e^{-\mu(b-c)}\\
%&=& \min\{e^{-\mu(b-c)},1\}\\
%&=& P(X>a)P(X>b)~.
%\eea
%\end{proof}

\begin{proof}[Proof of Proposition~\ref{prop:heavy_light_mean}]
Let $X$ be a random variable with a distribution that is heavy-everywhere. Consider any $x>0$. Using the property of being heavy-everywhere, we have
\bea
P(X>nx) &=& P(X>nx, X>x)\\
&=& P(X>nx|X>x)P(X>x)\\
& \geq & P(X> (n-1)x ) P(X > x)\nonumber\\
&\geq& P(X > (n-2)x) P(X>x) P(X>x)\nonumber\\
&\vdots& \nonumber\\
&\geq& P(X>x)^n~.
\eea
Now consider i.i.d. random variables $X_1,\ldots,X_n$ drawn from this distribution. The expected value of their minimum is given by
\bea
E[\min\{X_1,\ldots,X_n\}] &=&  \int P(X_1>x,\ldots,X_n>x) d\mu(x)\nonumber\\
&=& \int P(X_1>x)\cdots P(	X_n>x) d\mu(x)\nonumber\\
&=& \int P(X>x)^n d\mu(x)\nonumber\\
&\leq& \int P(X>nx) d\mu(x)\nonumber\\
&=& \frac{1}{n} E[X]~.
\eea

If the distribution is light-everywhere, then each of the inequalities in the entire proof above are flipped, leading to the result
\bea
E[\min\{X_1,\ldots,X_n\}] &\geq & \frac{1}{n} E[X]~.
\eea
\end{proof}

\begin{proof}[Proof of Proposition~\ref{prop:heavy_mix_is_heavy}]
Suppose $X$ is drawn from a mixture of $L$ independent random variables $X_1,\ldots,X_L$ for some $L\geq 1$ whose (marginal) distributions satisfy the conditions stated in the proposition. In particular, suppose $X$ takes value $X_i$ with probability $p_i \geq 0$ (with $\sum_{i=1}^{L} p_i = 1$). Then
\bea
P(X>a+b)&=& \sum_{i=1}^{L} p_i  P(X_i>a+b) \nonumber\\
&\geq & \sum_{i=1}^{L} p_i P(X_i>a)P(X_i>b) \nonumber\\
& = & \sum_{j=1}^{L} \sum_{i=1}^{L} p_i p_j P(X_i>a)P(X_i>b) \nonumber\\
& = & \left(\sum_{i=1}^{L} p_i P(X_i>a) \right) \left(\sum_{j=1}^{L} p_j P(X_j>b) \right)\nonumber\\
&& + \frac{1}{2} \sum_{j=1}^{L} \sum_{i=1}^{L} p_i p_j (P(X_i>a)-P(X_j>a))(P(X_i>b)-P(X_j>b)) \nonumber\\
&\geq & \left(\sum_{i=1}^{L} p_i P(X_i>a) \right) \left(\sum_{j=1}^{L} p_j P(X_j>b) \right)\label{eq:heavy_mix_0}\\
&=& P(X>a)P(x>b)~,
\eea
where~\eqref{eq:heavy_mix_0} is a result of the assumption that $P(X_i>a)\geq P(X_j>a) \Rightarrow P(X_i>b)\geq P(X_j>b)$.
\end{proof}

\begin{proof}[Proof of Proposition~\ref{prop:heavy_examples}]
Let $X$ be a random variable drawn from the distribution under consideration.
\begin{enumerate}
\item A mixture of a finite number of independently drawn exponential distributions.\\
The exponential distribution trivially satisfies~\eqref{eq:heavier_definition} and hence is heavy-everywhere. Furthermore, if $X_i$ and $X_j$ are exponentially distributed with rates $\mu_i$ and $\mu_j$, \[P(X_i>a)>P(X_j>a) \Rightarrow e^{-\mu_i a}>e^{-\mu_j a} \Rightarrow -\mu_i > -\mu_j \Rightarrow P(X_i>b)\geq P(X_j>b)~.\] This allows us to apply Prop.~\ref{prop:heavy_mix_is_heavy}, giving the desired result.
\item A Weibull distribution with scale parameter smaller than $1$.\\
The Weibull distribution has a complementary c.d.f.
\[
P(X>x) = e^{-\left(x/\lambda\right)^k}~.
\]
For $k\in(0,1]$, and for any $a,b>0$, we know that
\bea 
(a+b)^k &\leq& a^k + b^k\\
\Rightarrow -\left(\frac{a+b}{\lambda}\right)^k &\geq& -\left(\frac{a}{\lambda}\right)^k-\left(\frac{b}{\lambda}\right)^k\\
\Rightarrow e^{-\left(\frac{a+b}{\lambda}\right)^k} &\geq& e^{-\left(\frac{a}{\lambda}\right)^k}e^{-\left(\frac{b}{\lambda}\right)^k}\\
\Rightarrow P(X>a+b)&\geq& P(X>a)P(X>b)~.
\eea
\end{enumerate}
\end{proof}

\begin{proof}[Proof of Proposition~\ref{prop:light_sum_is_light}]
Let $X_1$ and $X_2$ be independent random variables whose (marginal) distributions are light-everywhere. %Let $Z$ be a Bernoulli random variable that takes a value $0$ with a probability $p$ and $1$ with a probability $(1-p)$, independent of $X_1$ and $X_2$. Define $X$ as:\[
%X = \begin{cases} X_1 & \textrm{if }Z=0\\X_2&\textrm{if }Z=1\end{cases}.\]
Let $X=X_1+X_2$, Then, 
\bea
P(X>a+b|X>b) &=& P(X_1+X_2>a+b,X_2>b|X_1+X_2>b)\nonumber\\
&&+P(X_1+X_2>a+b,X_2\leq b|X_1+X_2>b)\\
&=& P(X_1+X_2>a+b|X_2>b, X_1+X_2>b) P(X_2>b|X_1+X_2>b)\nonumber\\
&&+P(X_1+X_2>a+b | X_2\leq b, X_1+X_2>b) P(X_2 \leq b | X_1+X_2 >b).\label{eq:sum_light_proof0}
\eea
Now,
\bea
P(X_1+X_2>a+b|X_2>b,X_1+X_2>b) &=& P(X_1+X_2>a+b|X_2>b)\\
&=& P(X_2>a+b-X_1|X_2>b)\\
&\leq & P(X_2>a-X_1) \label{eq:sum_light_proof1}\\
&\leq & P(X_1+X_2>a)~,
\eea
where the inequality~\eqref{eq:sum_light_proof1} utilizes the light-everywhere property of the distribution of $X_2$.
Also,
\bea
P(X_1+X_2>a+b|X_2 \leq b, X_1+X_2>b) &=& P(X_1 >a+b-X_2 | X_2\leq b, X_1 > b-X_2)\\
&\leq& P(X_1 > a)\label{eq:sum_light_proof2}\\
&\leq& P(X_1 + X_2 > a)~,
\eea
where the inequality~\eqref{eq:sum_light_proof2} utilizes the light-everywhere property of the distribution of $X_1$. Putting it back together in~\eqref{eq:sum_light_proof0} we get
\bea
P(X\!>\!a\!+\!b|X\!>\!b) \!\!\!\!&\!\!\leq\!\!&\!\!\!\! P(X_1\!+\!X_2>a) P(X_2\!>\!b|X_1\!+\!X_2\!>\!b)\!+\! P(X_1\!+\!X_2>a) P(X_2 \!\leq\! b | X_1\!+\!X_2 \!>\!b)\\
&=&  P(X>a).
\eea
\end{proof}

\begin{proof}[Proof of Proposition~\ref{prop:light_examples}]
Let $X$ be a random variable drawn from the distribution under consideration.
\begin{enumerate}
\item For any $c>0$, the constant distribution with entire mass on $c$.\\
If $a \leq c$ then $P(X>a)=1$. If $a > c$ then $P(X > a+b)=0$. Thus the constant distribution satisfies~\eqref{eq:lighter_definition}.
\item An exponential distribution that is shifted by a positive constant.\\
The exponential distribution trivially satisfies~\eqref{eq:lighter_definition} and is light-everywhere. A constant is also light-everywhere as shown above. Applying Proposition~\ref{prop:light_sum_is_light}, we get the desired result.
\item The uniform distribution.\\
We first show that for every $M>0$. the uniform distribution on the interval $[0,M]$ is light-everywhere. If $a+b \geq M$ then $P(X> a+b)=0$, thus trivially satisfying~\eqref{eq:lighter_definition}. If $a+b<M$ then $P(X>a)=\frac{M-a}{M}$ and $P(X>a+b|X>b)=\frac{M-a-b}{M-b}$. Using the fact that $a \geq 0, b\geq 0$, some simple algebraic manipulations of these expressions lead to~\eqref{eq:lighter_definition}. Since a constant is light-everywhere, Proposition~\ref{prop:light_sum_is_light}  completes the result.
\item For any pair of non-negative constants $c_1$ and $c_2$ with $2c_1 > c_2>c_1$, a distribution with its support comprising only the two constants $c_1$ and $c_2$.\\
If $a+b\geq c_2$ then $P(X>a+b)=0$. If $a<c_1$ then $P(X>a)=1$. Finally, if $a\geq c_1$ and $a+b<c_2$ then the constraint of $2c_1>c_2$ implies $b < c_1$. Thus in this setting, $P(X>a+b|X>b)=P(X=c_2)=P(X>a)$.
\end{enumerate}
\end{proof}

\section{Proofs}\label{app:proofs}
We first present a brief description of the general proof technique we follow to obtain the analytical results, following which we provide the proofs of the individual results.

The general proof technique is depicted pictorially in Fig.~\ref{fig:proof_technique}. Consider two identical systems $S_1$ and $S_2$ with different redundant-requesting policies. Suppose we wish to prove that the redundant-requesting policy of system $S_2$ leads to a lower latency as compared to the redundant-requesting policy of system $S_1$. To this end we first construct two new hypothetical systems $T_1$ and $T_2$. The construction is such that the performance of system $T_1$ is statistically identical or better than $S_1$, and that of $T_2$ is statistically identical or worse than $S_2$. The two systems $T_1$ and $T_2$ are also \textit{coupled} in the following manner. The construction establishes a one-to-one correspondence between the $n$ servers of $T_1$ and the $n$ servers of $T_2$. Furthermore, it also establishes a one-to-one correspondence between the service events occurring in both systems, i.e., the completion of any job in $T_1$ is associated to the completion of a unique job in $T_2$ and vice versa. The same sequence of arrivals is applied to both systems. 

Such a coupling facilitates an apples-to-apples comparison between the two systems. We exploit this and show that at any point in time, system $T_2$ is in a better state than system $T_1$. Putting it all together, it implies that system $S_2$ is better than system $S_1$. 

Most interestingly, this technique allows us to handle arbitrary arrival sequences. Furthermore, it does not restrict the results to the (asymptotic) setting when the system is in steady state, but allows the results to be applicable to any interval of time.

\begin{figure*}
\centering
\includegraphics[width=.75\textwidth]{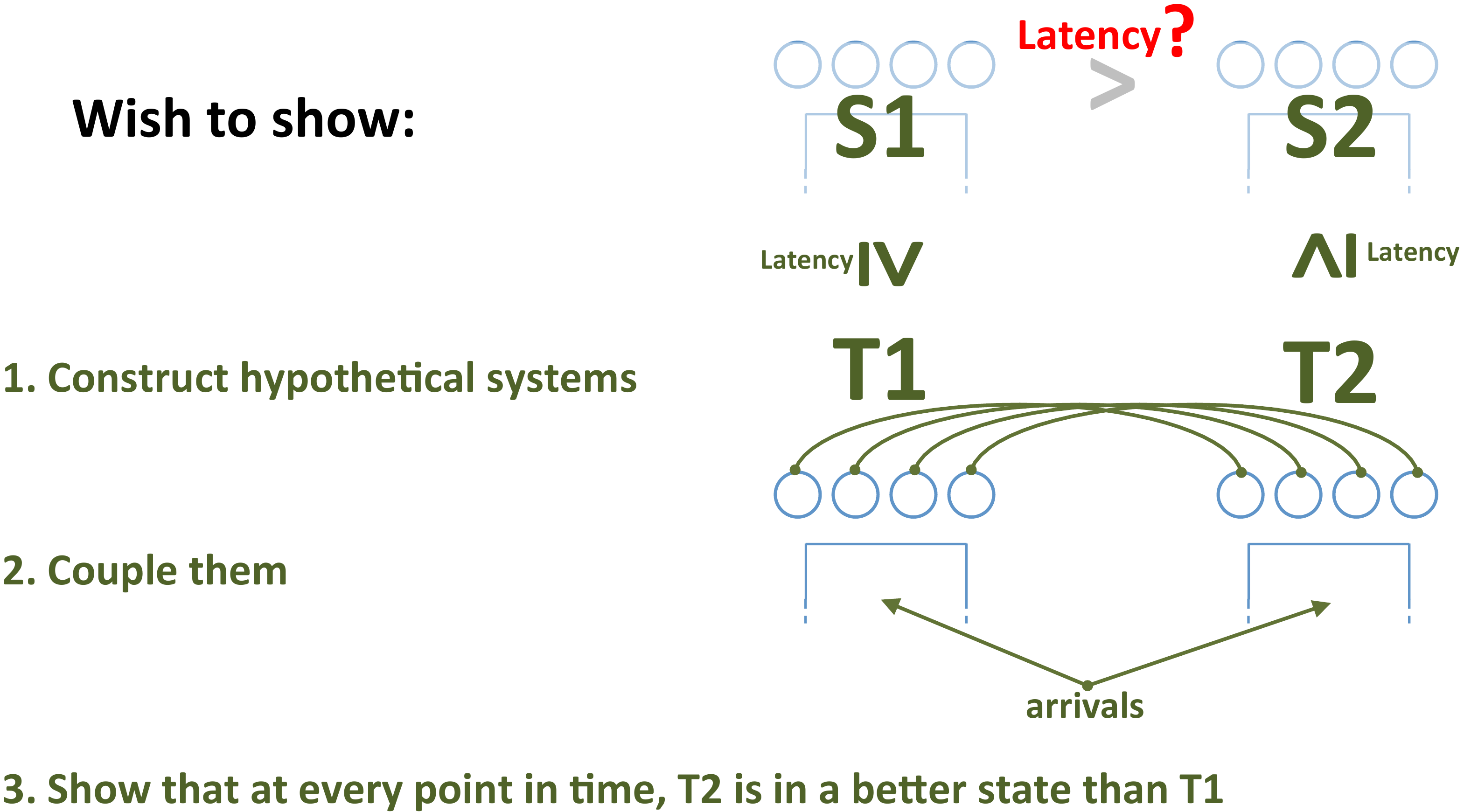}\\~
\caption{A pictorial depiction of the general proof technique followed in this paper.}
\label{fig:proof_technique}
\end{figure*}

We now provide proofs of the analytical results presented in the paper.

\begin{IEEEproof}[\textbf{Proof of Theorem~\ref{thm:rep_flood}} (centralized, memoryless service, no removal cost, $k=1$)]\label{proof:rep_flood}
Consider two systems, system $S_1$ with request-degree $\request_1$ and system $S_2$ with request-degree $\request_2~(>\request_1)$, both having system parameters $(n,\ k=1)$, the same arrival process, and the same rate of service. In the proof, we shall construct two new hypothetical systems $T_1$ and $T_2$ such that the statistics of $T_1$ are identical to $S_1$, and the statistics of $T_2$ are identical to $S_2$. We shall then show that system $T_2$ outperforms system $T_1$, and conclude that $S_2$ outperforms $S_1$.

The new system $T_1$ is defined as follows. The system $T_1$ is also associated to parameters $(n,\ k=1)$, has the same arrival and service processes as $S_1$, and follows the scheduling protocol described in Algorithm~\ref{alg:redundant_requests} with request-degree $\request_1$. However, after every service-event, we perform a specific permutation of the $n$ servers. Since the $n$ servers have independent and memoryless service time distributions with identical rates, the system $T_1$ remains statistically identical to $S_1$. In particular, the two systems $T_1$ and $S_1$ have identical distributions of the latency and buffer occupancy. The specific permutation applied is as follows. At any point in time, consider denoting the $n$ servers by indices `1',$\ldots$,`n'. Upon completion of any job at any server, the servers are permuted such that the busy servers have the lowest indices and the idle servers have the higher indices. In a similar manner, we construct $T_2$ to be a system identical to $S_2$, but again permuting the servers in $T_2$ after every job completion such that the busy servers have the lowest indices. Thus $T_2$ is statistically identical to $S_2$.

In the system under consideration, at any point in time, there are $(n+1)$ processes simultaneously going on: the arrival process and the processes at the $n$ servers. %TODO: Give more details on what we mean by the previous sentence. Also say that even if a server is idle, we can assume that the memoryless process in the server keeps going on.
The assumption of memoryless service times allows us to assume that a (fictitious) service process continues to execute even in an idle server, although no job is counted as served upon completion of the process. Let us call the completion of any of these processes as an \textit{event}. In this proof, we assume the occurrence of any arbitrary sequence of events, and evaluate the performance of systems $T_1$ and $T_2$ under this sequence of events. Since the arrivals into the system and the memoryless processes at the servers are all independent of the state of the system, we can assume \textit{the same} sequence of events to occur in the two systems.

We begin by showing that under an identical sequence of events (the arrivals and server completions) in systems $T_1$ and $T_2$, the number of batches remaining to be completely served in $T_2$ at any point of time is no more than number of batches remaining in $T_1$ at that time. Without loss of generality, we shall prove this statement only at times immediately following an event, since the systems do not change state between any two consecutive events. With some abuse of notation, for $z \in \{0,1,2,\ldots\}$, we shall use the term ``time $z$'' to denote the time immediately following the $z\supth$ event.
%For any $z \in \{0,1,2,\ldots\}$, the time immediately following the $z\supth$ event  We shall denote the time by a discrete variable $z \in \{0,1,2,\ldots\}$ corresponding to the $0\supth$, $1^\textrm{st}$, $2^\textrm{nd}$,$\ldots$ events respectively. Under this notation, ``time $z$'' represents the time immediately after the $z\supth$ event.

Assume that the two systems begin in identical states at time $0$. For system $T_i~(i\in\{1,2\})$, let $b_i(z)$ denote the number of batches remaining in system $T_i$ at time $z$. The proof proceeds via induction on $z$. The induction hypothesis is that at any time $z$, we have $b_1(z) \geq b_2(z)$. Since the two systems begin in identical states, $b_1(0)=b_2(0)$. Now suppose the induction hypothesis is satisfied at time $(z-1)$. We shall now show that it is satisfied at time $z$ as well.

Suppose the $z\supth$ event is the arrival of a new batch. Then
\bea b_1(z) &=& b_1(z-1)+1\\
&\geq& b_2(z-1)+1\label{eq:rep_proof_a} \\
&=&b_2(z)
\eea
where~\eqref{eq:rep_proof_a} follows from the induction hypothesis. Thus, the hypothesis is satisfied at time $z$. 

Now suppose the $z\supth$ event is the completion of the exponential timer of one of the $n$ servers (in both the systems). We first consider the case $b_1(z-1) \geq b_2(z-1) + 1$. Since the completion of the timer at a server can lead to the completion of the service of at most one batch, it follows that $b_1(z) \geq b_2(z)$ in this case. Now consider the case $b_1(z-1) = b_2(z-1)$. Since $k=1$, the number of servers occupied in system $T_i~(i\in\{1,2\})$ at time $(z-1)$ is equal to $\min\{\request_i b_i(z-1),\ n\}$. Furthermore, from the construction of systems $T_1$ and $T_2$ described above (recall the permutation of servers), it must be that the first $\min\{\request_i b_i(z-1),\ n\}$ servers are occupied at time $(z-1)$ in system $T_i$. Thus, since $\request_1<\request_2$ and $b_1(z-1) = b_2(z-1)$, the set of servers occupied at time $(z-1)$ in $T_1$ is a subset of the servers occupied in $T_2$. Now, since $k=1$, an event at a server triggers the completion of service of a batch if and only if that server was not idle. Thus, if this event leads to the completion of service of a batch in $T_1$, it also leads to the completion of service of a batch in system $T_2$. It follows that $b_1(z) \geq b_2(z)$. We have thus shown that at any point in time, the number of batches remaining in system $T_2$ is no more than that under system $T_1$.

The arguments above show that the distribution of the number of batches remaining in $T_1$ dominates that in $T_2$: with $B_1$ and $B_2$ denoting the number of batches in the system $T_1$ and $T_2$ respectively under steady state, $P(B_1 > x) \geq P(B_2 > x)$ for all $x\geq 0$. Since the average latency is proportional to the average system occupancy, it follows that the latency faced by a batch on an average in system $T_2$ is no more than that in $T_1$. These properties carry over to $S_1$ and $S_2$ since the statistics of $S_1$ and $S_2$ are identical to those of $T_1$ and $T_2$ respectively.

From arguments identical to the above, it follows that having a request degree of $n$ for each batch minimizes the average latency as compared to any other redundant requesting policy, including ones where a different request degree may be chosen (adaptively) for different batches.

Finally, we show that if $T_1$ employs a fixed request-degree $\request<n$ for \textit{all} batches, and $T_2$ employs $r=n$ for all batches, then the average latency under $T_2$ is strictly smaller. At any given time, there is a non-zero probability of the occurrence of a sequence of service-events that empty system $T_1$ (which also results in $T_2$ getting emptied). Now, upon arrival of a batch, this new batch is served in $\request<n$ servers of $T_1$ and in all $n$ servers of $T_2$, and hence there is a strictly positive probability that the batch completes service in $T_2$ before it completes service in $T_1$ and also before a new batch arrives. This event results in $b_2(\cdot)<b_1(\cdot)$, and since this event occurs with a non-zero probability, we can draw the desired conclusion. 
\end{IEEEproof}

\begin{IEEEproof}[\textbf{Proof of Theorem~\ref{thm:memoryless_general_k}} (centralized, memoryless service, no removal cost, general $k$)]\label{proof:memoryless_general_k}
Consider two systems, system $S_1$ with an arbitrary redundant-requesting policy and  system $S_2$ with request-degree $n$, both having system parameters $(n,\ k)$, the same arrival process, and the same rate of service. In the proof, we shall construct two new systems $T_1$ and $T_2$ such that the statistics of $T_1$ are identical to $S_1$, and the statistics of $T_2$ are identical to $S_2$. We shall then show that system $T_2$ outperforms system $T_1$, and conclude that $S_2$ outperforms $S_1$.

In either system, at any point in time, there are $(n+1)$ processes simultaneously going on: the arrival process and the processes at the $n$ servers. The assumption of memoryless service times allows us to assume that a (fictitious) service process continues to execute even in an idle server, although no job is counted as served upon completion of the process. 
Let us term the completion of any of these $(n+1)$ timers as the an \textit{event}. In this proof, we assume the occurrence of any arbitrary sequence of events, and evaluate the performance of systems $T_1$ and $T_2$ under this sequence of events. Since the arrivals into the system and the memoryless processes at the servers are all independent of the state of the system, we can assume \textit{the same} sequence of events to occur in the two systems.

We shall now show that under an identical sequence of events (arrivals and server completions) in $T_1$ and $T_2$, the number of batches remaining in system $T_1$ is at least as much as that in $T_2$ at any given time. Without loss of generality, we shall prove this statement only at times immediately following an event, since the states of the systems do not change in between any two events. Abusing some notation, for $z \in \{0,1,2,\ldots\}$, we shall use the term ``time $z$'' to denote the time immediately following the $z\supth$ event. 

Assume that the two systems begin in the same state at time $z=0$. For system $T_i~(i\in\{1,2\})$, let $b_i(z)$ denote the number of batches remaining in system $T_i$ at time $z$. The proof proceeds via induction on the time $z$. The induction hypothesis is that at any time $z$:
\begin{enumerate}[(a)]
\item $b_1(z) \geq b_2(z)$, and
\item for any $z'>z$, if there are no arrivals between time $z$ and $z'$ (including at time $z'$), then $b_1(z') \geq b_2(z')$.
\end{enumerate}
The hypotheses are clearly true at $z=0$, when the two systems are in the same state. Now, let us consider them to be true for time $z~(\geq 0)$. Suppose the next event occurs at time $(z+1)$. We need to show that the hypotheses are true even after this event at time $(z+1)$.

First suppose the event was the completion of an exponential-timer at one of the $n$ servers. Then there has been no arrival between times $z$ and $(z+1)$. This allows us to apply hypothesis (b) at time $z$ with $z' = z+1$, which implies the satisfaction of both the hypotheses at time $(z+1)$. 

Now suppose the event at time $(z+1)$ is the arrival of a new batch. Then, hypothesis (a) is satisfied at time $(z+1)$ since $b_1(z+1)=b_1(z)+1\geq b_2(z)+1=b_2(z+1)$. We now show that hypothesis (b) is also satisfied. Consider any sequence of server-events, and any time $z' > z+1$ such that there were no further arrivals between times $(z+1)$ and $z'$.

Let $a_1(z')$ and $a_2(z')$ be the number of batches remaining in the two systems at time $z'$ \textit{if the new batch had not arrived} but the sequence of server-events was the same as before. From hypothesis (b) at time $z$, we know that $a_1(z') \geq a_2(z')$. Also note that the scheduling protocol described in Algorithm~\ref{alg:redundant_requests} gives priority to the batch that had arrived earliest, and as a consequence, a server serves a job from the new batch only when it cannot serve any other batch. It follows that under any sequence of server-events, for $i\in\{1,2\}$, $b_i(z') = a_i(z')+1$ if $k$ jobs of the new batch have not completed service in $T_i$, else $b_i(z') = a_i(z')$. When $ b_1(z')=a_1(z')+1$, it follows that $ b_1(z')=a_1(z')+1 \geq a_2(z') + 1 \geq b_2(z')$. It thus remains to show that $b_1(z')=a_1(z') \Rightarrow b_2(z')  \leq b_1(z')$. The condition $b_1(z')=a_1(z')$ implies that $k$ jobs of the new batch have completed service in system $T_1$ at or before time $z'$. Let $z_1,\ldots,z_k$ ($z_1<\ldots<z_k\leq z'$) be the events when the $k$ jobs of the new batch are served in system $T_1$. Then, at these times, the corresponding servers must have been idle in system $T_1$ if the new batch had not arrived.

Consider another sequence of events that is identical to that discussed above, but excludes the server-events that happened at times $z_1,\ldots,z_k$, and also excludes the arrival at time $(z+1)$. Let $c_i(z')$ denote the number of batches remaining in this situation at time $z'$. From the arguments above, we get $c_1(z')=a_1(z')$. From the second hypothesis, we also have $c_2(z') \leq c_1(z')$. Thus we already have $c_2(z') \leq c_1(z') = a_1(z') = b_1(z')$, and hence for our goal of showing $b_2(z')\leq b_1(z')$, it now suffices to show that $b_2(z')\leq c_2(z')$. 

If $b_2(z')=0$ then we automatically have $b_2(z') \leq b_1(z')$ and there is nothing left to show. Thus, we consider the case $b_2(z')>0$, i.e., system $T_2$ is non-empty at time $z'$. We shall now see how to count the number of batches in any system at time $z'$, under the condition that there were \textit{no arrivals} between time $(z+1)$ and $z'$. Consider $(n+1)$ counters: one counter each for the $n$ servers and one `global' counter. At time $(z+1)$, let the value of the counter of any server be equal to the number of jobs that this server has finished serving from the batches that are still remaining in the system. Let the value of the global counter be $0$ at this time. Now, whenever a server-event occurs, add $1$ to the counter associated to that server, irrespective of whether the server had a job or not. Whenever the counters of any $k$ servers become greater than zero, add $1$ to the global counter, and subtract $1$ from the counters of these $k$ servers. One can see that in this process, the value of the global counter at any time gives the number of batches that have finished service since the time we started counting. With this in mind, we shall compare the sequence of events that includes the events at $z_1,\ldots,z_k$ to that which excludes these events. Since the events $z_1,\ldots,z_k$ must correspond to events at $k$ \textit{distinct} servers, the service-events at $z_1,\ldots,z_k$ cause the global counter of system $T_2$ to increase by one. Since $T_2$ also had one additional arrival as compared to the system of $c_2(\cdot)$, it must be that $b_2(z')=c_2(z')$. Putting the pieces together, we get that the number of batches served in $T_2$ at any time is at least as much as that served in $T_1$ at any time.

Since the average latency is proportional to the average system occupancy, it follows that the latency faced by a batch on an average in system $T_2$ is smaller than that in $T_1$. These properties carry over to $S_1$ and $S_2$ since the statistics of $S_1$ and $S_2$ are identical to those of $T_1$ and $T_2$ respectively.

Finally, we show that if $T_1$ employs a fixed request-degree $\request<n$ for \textit{all} batches, and $T_2$ employs $r=n$ for all batches, then the average latency under $T_2$ is strictly smaller. At any given time, there is a non-zero probability of the occurrence of a sequence of service-events that empty system $T_1$ (which also results in $T_2$ getting emptied). Now, upon arrival of a batch, this new batch is served in $\request<n$ servers of $T_1$ and in all $n$ servers of $T_2$, and hence there is a strictly positive probability that the batch completes service in $T_2$ before it completes service in $T_1$ and also before a new batch arrives. This event results in $b_2(\cdot)<b_1(\cdot)$, and since this event occurs with a non-zero probability, we can draw the desired conclusion.
% The result that we just proved above also implies that the stability region of a system with $\request=n$ contains the stability region of a system with a smaller request-degree. In particular, the stability regions of $T_1$ and $T_2$ are subsets of the stability region of a system without any redundancy in the requests (i.e., with $\request=k$). Since the system (is assumed to be) stable, the state where both systems are empty has a probability of occurrence that is strictly bounded away from zero in the steady state.  Consider the next arrival, whose $\request_i$ jobs are assigned (without loss of generality) to the first $\request_i$ servers under system $S_i$. Suppose this arrival occurs at time $z$. Then $b_2(z)=b_1(z)=1$. Suppose the next few events are completion of the exponential timers at different servers. Then since $n > \request$, there is a probability strictly bounded away from zero that some $k$ of the first $\request_2$ servers first complete service, and at least one of these servers is in the set $\{\request_1+1,\ldots,\request_2\}$. This takes the system to a state with $\{b_1(\cdot)=1,\ b_2(\cdot)=0\}$ with a probability strictly bounded away from zero.  Thus, the fraction of time $f$ in which $b_1(\cdot) > b_2(\cdot)$ is lower bounded by the product of the aforementioned probabilities, and is hence strictly bounded away from zero.
Thus, the distribution of the system occupancy in $T_2$ is strictly dominated by that of $T_1$.
\end{IEEEproof}

\begin{proof}[\textbf{Proof of Theorem~\ref{thm:heavier_flood}} (centralized, heavy-everywhere service, no removal cost, $k=1$, high load)]\label{proof:heavier_flood}
Consider two systems, system $S_1$ with some arbitrary redundant-requesting policy, and system $S_2$ with request-degree $n$ for all batches. We shall now construct two new hypothetical systems $T_1$ and $T_2$ such that $T_1$ is statistically identical to $S_1$ and $T_2$ is worse than $S_2$, and show that the performance of $T_1$ is worse than that of $T_2$.

The two new systems $T_1$ and $T_2$ are constructed as follows. Both systems have the same parameters $n$ and $k=1$, and retain the redundant-requesting policies of $S_1$ and $S_2$ respectively. The service-time distribution in $T_1$ is identical to that in $S_1$. On the other hand, we shall make the service time distribution of $T_2$ \textit{worse} than that of $S_2$ in the manner described below. 

Let us fix some arbitrary one-to-one correspondence between the $n$ servers of system $T_1$ and the $n$ servers of system $T_2$. Consider any point in time when a server in system $T_2$ is just beginning the service of a job. Let $X$ denote the random variable corresponding to this service time. Let $P_H$ denote the law associated to the heavy-everywhere service-time distribution under consideration. Since the systems operate at 100\% server utilization, the corresponding server in $T_1$ is not idle at this point in time and is serving some job. When this server in system $T_2$ begins service, suppose the job in the corresponding server of system $T_1$ began to be serviced $t>0$ units of time ago. Then we modify the distribution of $X$ and let it follow the law 
\[P(X>x) = P_H(X>x+t|X>t)\quad \forall x \geq 0~.\]
Since the distribution $P_H$ is heavy-everywhere~\eqref{eq:heavier_definition}, the service under system $T_2$ is no better than that under $S_2$. As a result of the construction above, whenever a job begins to be processed in system $T_2$, it has a service-time distribution that is identical to the distribution of the service-time of the job in the corresponding server in $T_1$. We couple the servers even further by assuming whenever a server in $T_2$ begins a new job, the time taken for this job to be completed is \textit{identical} to that taken for the job in the corresponding server in $T_1$ (unless, of course, some other job of the batch completes service first and this job is removed). We also feed an identical sequence of arrivals to the two systems $T_1$ and $T_2$. This completes the construction of the two systems $T_1$ and $T_2$.

Note that the aforementioned coupling of service-times between corresponding servers of systems $T_1$ and $T_2$ only takes place when the server in $T_2$ begins a job. The case when a server of $T_1$ begins serving a new job when the corresponding server of $T_2$ is already serving a job is not accounted for. The induction hypothesis below handles such situations.

We start at any point in time when the two systems are in an identical state, and show that the average latency faced by the batches in $T_2$ from then on is no larger than that faced by batches in $T_1$. We shall now show the following two properties via an induction on time:
\begin{enumerate}
\item[a)] At any point in time, the number of batches in system $T_2$ is no more than the number of batches in system $T_1$.
\item[b)] At any point in time, if a server in system $T_1$ begins service of a job, the corresponding server in $T_2$ also begins service of some job. 
\end{enumerate}
Part (b) of the hypothesis ensures that the service times of the jobs in corresponding servers of systems $T_1$ and $T_2$ are always identical (via the construction above).

As mentioned previously, let us start at any point in time when the two systems are in an identical state. Since the systems are in an identical state, both hypotheses hold true at this time. Without loss of generality, we shall now consider only the times immediately following an event in either system, where an event is defined as an arrival of a batch or the completion of processing by a server. First consider any time that immediately follows an arrival. By our induction hypothesis, just before the arrival, the number of batches in $T_2$ was no more than that in $T_1$. The arrival only increases the number of batches in both systems by $1$, and hence induction hypothesis (a) still stands. Under a 100\% server utilization, an arrival does not trigger the beginning of a service in either system. Thus, hypothesis (b) continues to hold. Let us now consider an event where a server completes processing a job. Due to hypothesis (b), the service times at corresponding servers in the two systems were coupled. As a result, the next service completes at the same time in corresponding servers of both systems. This reduces the number of batches in both systems by one, thus continuing to satisfy hypothesis (a). Furthermore, since we have assumed a 100\% utilization of the servers, there is at least one batch waiting in the buffer in both the systems. In system $T_2$, since we had $k=1$, $r=n$ and no removal cost, at any given time each of the $n$ servers in system $T_2$ will be serving jobs of the same batch. Thus jobs in all the servers of $T_2$ are removed from the system, and are replaced by (new) jobs of the next batch. As a result, upon any service-event, each of the servers in $T_2$ begin serving new jobs, thus satisfying hypothesis (b). Due to the specific construction of the two systems, the service times of these new jobs in the servers of $T_2$ are identical to those of jobs in corresponding servers of $T_1$.

This completes the proof of the induction hypothesis, and in particular that the number of batches in $T_2$ at any time is no more than the number of batches in system $T_1$. The fact that the average latency is proportional to the average number of batches in the system implies that the average latency in system $T_1$ is no smaller than in $T_2$. Finally, the constructions of the two systems $T_1$ and $T_2$ ensured that system $T_2$ is worse than $S_2$, and system $T_1$ is statistically identical to $S_1$, thus leading to the desired result.

Finally, suppose system $T_1$ employs a fixed request-degree $r<n$ for all batches. Further suppose that the heavy-everywhere distribution is such that~\eqref{eq:heavier_definition} holds with a strict inequality for a set of events that have a probability bounded away from zero. Under this setting, the aforementioned construction is such that system $T_2$ is worse than system $S_2$ by a non-trivial amount, and as a result, the average latency in system $S_2$ is strictly smaller than that of $S_1$.
\end{proof}

\begin{proof}[\textbf{Proof of Theorem~\ref{thm:lighter_noflood}} (centralized, light-everywhere service, any removal cost, $k=1$, high load)]\label{proof:lighter_noflood}
Consider two systems, system $S_1$ with some arbitrary redundant-requesting policy, and system $S_2$ with request-degree $r=k=1$ for all batches. We shall now construct two new hypothetical systems $T_1$ and $T_2$ such that $T_1$ is statistically identical to $S_1$ and $T_2$ is worse than $S_2$, and show that the performance of $T_1$ is worse than that of $T_2$.

The two new systems $T_1$ and $T_2$ are constructed as follows. Both systems have the same parameters $n$ and $k=1$, and retain the redundant-requesting policies of $S_1$ and $S_2$ respectively. The service-time distribution in $T_2$ is identical to that in $S_2$. On the other hand, we shall make the service time distribution of $T_1$ \textit{better} than that of $S_1$ in the manner described below.

Let us fix some arbitrary one-to-one correspondence between the $n$ servers of system $T_1$ and the $n$ servers of system $T_2$. Consider any point in time when a server in system $T_1$ is just beginning the service of a job. Let $X$ denote the random variable corresponding to this service time. Let $P_L$ denote the law associated to the heavy-everywhere service-time distribution under consideration. Since the systems operate at 100\% server utilization, the corresponding server in $T_2$ is not idle at this point in time and is serving some job. When this server in system $T_2$ begins service, suppose the job in the corresponding server of system $T_1$ began to be served $t>0$ units of time ago. Then we modify the distribution of $X$ and let it follow the law 
\[P(X>x) = P_L(X>x+t|X>t)\quad \forall x \geq 0~.\]
Since the distribution $P_L$ is light-everywhere~\eqref{eq:lighter_definition}, the service under system $T_1$ is no better than that under $S_1$. As a result of the construction above, whenever a job begins to be processed in system $T_1$, it has a service-time distribution that is identical to the distribution of the service-time of the job in the corresponding server in $T_2$. We couple the servers even further by assuming whenever a server in $T_1$ begins a new job, the time taken for this job to be completed is \textit{identical} to that taken for the job in the corresponding server in $T_2$ (unless, of course, some other job of the batch completes service first and this job is removed). We also feed an identical sequence of arrivals to the two systems $T_1$ and $T_2$. This completes the construction of the two systems $T_1$ and $T_2$.

Note that the aforementioned coupling of service-times between corresponding servers of systems $T_1$ and $T_2$ only takes place when the server in $T_1$ begins a job. The case when a server of $T_2$ begins serving a new job when the corresponding server of $T_1$ is already serving a job is not accounted for. The induction hypothesis below handles such situations.

We start at any point in time when the two systems are in an identical state, and show that the average latency faced by the batches in $T_2$ from then on is no more than that faced by batches in $T_1$. We shall now show the following two properties via an induction on time:
\begin{enumerate}
\item[a)] At any point in time, the number of batches in system $T_2$ is no more than the number of batches in system $T_1$.
\item[b)] At any point in time, if a server in system $T_2$ begins service of a job, the corresponding server in $T_1$ also begins service of some job. 
\end{enumerate}
Part (b) of the hypothesis ensures that the service times of the jobs in corresponding servers of systems $T_1$ and $T_2$ are always identical (via the construction above).

As mentioned previously, let us start at any point in time when the two systems are in an identical state. Since the systems are in an identical state, both hypotheses hold true at this time. Without loss of generality, we shall now consider only the times immediately following an event in either system, where an event is defined as an arrival of a batch or the completion of processing at server. First consider any time that immediately follows an arrival. By our induction hypothesis, just before the arrival, the number of batches in $T_2$ was no more than that in $T_1$. The arrival only increases the number of batches in both systems by $1$, and hence induction hypothesis (a) still stands. Under a 100\% server utilization, an arrival does not trigger the beginning of a service in either system. Thus, hypothesis (b) continues to hold. Let us now consider an event where a server completes processing a job. Due to hypothesis (b), the service times at corresponding servers in the two systems were coupled. As a result, the next service completes at the same time in corresponding servers of both systems. This reduces the number of batches in both systems by one, thus continuing to satisfy hypothesis (a). Furthermore, since we have assumed a 100\% utilization of the servers, there is at least one batch waiting in the buffer in both the systems. In system $T_2$, since we had $k=1$ and $r = k=1$, only this server begins serving a new job, while all remaining $(n-1)$ servers continue processing the jobs they already have. Now, the corresponding server in $T_1$ also had a server-event and begins serving a new job at this moment. Thus this satisfies hypothesis (b). Due to the specific construction of the two systems, the service times of these new jobs in the servers of $T_2$ are identical to those of jobs in corresponding servers of $T_1$.

This completes the proof of the induction hypothesis, and in particular that the number of batches in $T_2$ at any time is no more than the number of batches in system $T_1$. The fact that the average latency is proportional to the average number of batches in the system implies that the average latency in system $T_1$ is no smaller than in $T_2$. Finally, the constructions of the two systems $T_1$ and $T_2$ ensured that system $T_2$ is worse than $S_2$, and system $T_1$ is statistically identical to $S_1$, thus leading to the desired result.

Finally, suppose system $T_1$ employs a fixed request-degree $r<n$ for all batches. Further suppose that the light-everywhere distribution is such that~\eqref{eq:lighter_definition} holds with a strict inequality for a set of events that have a probability bounded away from zero. Under this setting, the aforementioned construction is such that system $T_1$ is better than system $S_1$ by a non-trivial amount, and as a result, the average latency in system $S_2$ is strictly smaller than that of $S_1$.
\end{proof}

\begin{proof}[\textbf{Proof of Theorem~\ref{thm:removalCosts_memless}} (centralized, memoryless service, non-zero removal cost, $k=1$, high load)]\label{proof:removalCosts_memless}
The proof is identical to that of Theorem~\ref{thm:lighter_noflood}. The system $T_2$ is constructed to be `better' than system $S_2$ by assuming zero removal costs in $T_2$.
\end{proof}

\begin{proof}[\textbf{Proof of Theorem~\ref{thm:memoryless_general_k_distributed}} (distributed, memoryless service, no removal cost, general $k$)]\label{proof:memoryless_general_k_distributed}
In order to get to the desired results, we shall first compare a system with distributed buffers to an analogous system that has a central buffer. The scheduling policy in either setting needs to make two kinds of decisions: the number of redundant requests for each batch, and the precise set of servers to which these requests are assigned. Firstly, observe that under the redundant requesting policy of $\request=n$ for all batches, the choice of the $\request~(=n)$ servers to which the jobs are assigned does not require any decision to be made. The centralized and the distributed systems thus are identical in this case and hence have the same average latency. Secondly, we know from the results of Section~\ref{sec:analytical} that under all the arrival and service time distributions considered here, the choice of $\request=n$ for all batches is optimal under a centralized scheme. Thirdly, for any fixed redundant requesting policy, the average latency under the centralized scheme will be no more than the average latency under the distributed scheme. This is because the first-come first-served policy as described in Algorithm~\ref{alg:redundant_requests} minimizes average latency, and moreover, the policy of the system with a centralized buffer has more information (about the system) as compared to the one operating under distributed buffers. It thus follows that even in the case of distributed buffers, the average latency is minimized when the request-degree for each batch is $n$.
%for `Distributed(r=n) $=$ Centralized(r=n) $\geq $ Centralized(other policy) $\geq$ Distributed(other policy)' which proves the desired result.
\end{proof}

\begin{proof}[\textbf{Proof of Theorem~\ref{thm:heavy_distributed}} (distributed, heavy-everywhere service, no removal cost, $k=1$, high load)]\label{proof:heavier_noflood_distributed}
Identical to the proof of Theorem~\ref{thm:heavier_flood}.
\end{proof}

\begin{proof}[\textbf{Proof of Theorem~\ref{thm:lighter_noflood_distributed}} (distributed, light-everywhere service, any removal cost, $k=1$, high load)]\label{proof:lighter_noflood_distributed}
Identical to the proof of Theorem~\ref{thm:lighter_noflood}.
\end{proof}

\begin{proof}[\textbf{Proof of Theorem~\ref{thm:removalCosts_memless_distributed}} (distributed, memoryless service, non-zero removal cost, $k=1$, high load)]\label{proof:removalCosts_memless_distributed}
Identical to the proof of Theorem~\ref{thm:removalCosts_memless}.
\end{proof}

%\section{Simulations of Additional Scenarios}\label{app:more_simulations}
%In this section, we provide results from simulations of more general scenarios.

\end{document}